\pdfoutput=1
\documentclass[letterpaper]{article} % DO NOT CHANGE THIS
\usepackage{aaai24}  % DO NOT CHANGE THIS
\usepackage{times}  % DO NOT CHANGE THIS
\usepackage{helvet}  % DO NOT CHANGE THIS
\usepackage{courier}  % DO NOT CHANGE THIS
\usepackage[hyphens]{url}  % DO NOT CHANGE THIS
\usepackage{graphicx} % DO NOT CHANGE THIS
\usepackage{natbib}  % DO NOT CHANGE THIS AND DO NOT ADD ANY OPTIONS TO IT
\usepackage{caption} % DO NOT CHANGE THIS AND DO NOT ADD ANY OPTIONS TO IT
\usepackage{algorithm}
\usepackage{algorithmic}

\usepackage{paralist}
\usepackage{amsmath,amssymb,amsthm}
\usepackage{tikz}
\usetikzlibrary{calc,shapes,arrows}
\usepackage{thmtools,thm-restate}
\usepackage{xspace}
\usepackage{todonotes}

\usepackage{newfloat}
\usepackage{listings}
\DeclareCaptionStyle{ruled}{labelfont=normalfont,labelsep=colon,strut=off} % DO NOT CHANGE THIS
\floatstyle{ruled}
\newfloat{listing}{tb}{lst}{}
\floatname{listing}{Listing}
\newcommand{\m}[1]{\mathsf{#1}}
\newcommand{\mc}[1]{\mathcal{#1}}
\newcommand{\obar}[1]{\makebox[0pt]{$\phantom{#1}\overline{\phantom{#1}}$}#1}
\renewcommand{\vec}[1]{\obar{#1}} % vector notation
\newcommand{\inn}{\,{\in}\,} % narrow \in

\newcommand{\TT}{\mc T} % theory
\newcommand{\LL}{\mc L} % verification language
\newcommand{\CC}{\mc C} % set of all constraints
\newcommand{\BB}{\mc B} % a DDS
\newcommand{\NN}{\mc N} % an NFA
\renewcommand{\SS}{\mc S} % sorts
\newcommand{\PP}{\mc P} % predicates
\newcommand{\FF}{\mc F} % function symbols
\newcommand{\eval}[1]{[#1]_{\alpha}^M} % domain
\newcommand{\goto}[1]{\mathrel{\raisebox{-2pt}{$\xrightarrow{#1}$}}}
\newcommand{\update}{\mathit{update}} % update operation in CG construction
\newcommand{\constr}{\mathit{cstr}} % constraints in element of \Sigma
\newcommand{\trans}[1]{\widehat{#1}}%{\Delta_{#1}} % transition formula
\newcommand{\phiinit}{\phi_{I}}
\newcommand{\NFA}[1][\psi]{{\mc N}_{#1}}
\newcommand{\U}{\mathrel{\mathsf{U}}} % LTL operator
\newcommand{\G}{\mathsf{G}\xspace}
\newcommand{\F}{\mathsf{F}\xspace}
\newcommand{\X}{\mathsf{X}\xspace}

\newcommand{\hist}{H_\exists}
\newcommand{\writ}{\mathit{write}}
\newcommand{\inquotes}[1]{#1}
\newcommand{\last}{\lambda}
\newcommand{\eqn}{\,{=}\,}
\newcommand{\ovee}{\mathrel{\tikz{\node[circle, inner sep=0pt, scale=0.7,draw,line width=0.01mm]{$\vee$}}}}
\newcommand{\owedge}{\mathrel{\tikz{\node[circle, inner sep=0pt, scale=0.7,draw,line width=0.01mm]{$\wedge$}}}}
\newcommand{\LTLf}{LTL$_f$\xspace}

\renewcommand{\phi}{\varphi}
\newcommand{\mydds}[1][$\Sigma$]{{#1}-DMT\xspace}
\newcommand{\plainmydds}{DMT\xspace}
\newcommand{\DBDDS}{DB$_\Sigma$-DMT\xspace}

\newtheorem{theorem}{Theorem}

\newtheorem{remark}[theorem]{Remark}
\newtheorem{example}[theorem]{Example}
\theoremstyle{definition}
\newtheorem{definition}[theorem]{Definition}
\newtheorem{lemma}[theorem]{Lemma}
\newtheorem{corollary}[theorem]{Corollary}

\declaretheorem[numbered=no, name={Assumptions ($\star$)}]{assumptions}

\newcommand{\lemref}[1]{Lem.~\ref{lem:#1}}

\newcommand{\defref}[1]{Def.~\ref{def:#1}}

\newcommand{\corref}[1]{Cor.~\ref{cor:#1}}
\newcommand{\propref}[1]{Prop.~\ref{prop:#1}}

\newcommand{\thmref}[1]{Thm.~\ref{thm:#1}}
\newcommand{\remref}[1]{Rem.~\ref{rem:#1}}
\newcommand{\exaref}[1]{Ex.~\ref{exa:#1}}

\newcommand{\exasref}[2]{Exs.~\ref{exa:#1} and \ref{exa:#2}}

\tikzstyle{state}=[draw, circle, inner sep=1.5pt, line width=.7pt, scale=.6]
\tikzstyle{edge}=[draw, ->, line width=.5pt]
\tikzstyle{action}=[scale=.7]
\tikzstyle{caption}=[scale=.9]

\newcommand{\pcnoder}[2]{{#1}\nodepart{two}{#2}}

\newcommand{\dmttuple}{\langle \Sigma, V, I, T\rangle}

\newcommand{\theproperty}{finite data history\xspace}
\newcommand{\historyset}{data history\xspace}

\title{Linear-Time Verification of Data-Aware
Processes Modulo Theories\\ via Covers and Automata (Extended Version)}
\author {
    Alessandro Gianola\textsuperscript{\rm 1},
    Marco Montali\textsuperscript{\rm 2},
   Sarah Winkler\textsuperscript{\rm 2}
}
\affiliations {
    \textsuperscript{\rm 1}Instituto Superior Técnico, University of Lisbon, Portugal\\
    \textsuperscript{\rm 2}Free University of Bozen-Bolzano, Italy\\
    alessandro.gianola@tecnico.ulisboa.pt, montali@inf.unibz.it, winkler@inf.unibz.it
}
\begin{document}

\maketitle

\begin{abstract}
The need to model and analyse dynamic systems operating over complex data is ubiquitous in AI and neighboring areas, in particular business process management. Analysing such data-aware systems is a notoriously difficult problem, as they are intrinsically infinite-state. Existing approaches work for specific datatypes, and/or limit themselves to the verification of safety properties.
In this paper, we lift both such limitations, studying for the first time linear-time verification for so-called \emph{data-aware processes modulo theories} (DMTs), from the foundational and practical point of view. The DMT model is very general, as it supports processes operating over variables that can store arbitrary types of data, ranging over infinite domains and equipped with domain-specific predicates. 
%A notable instance of this framework, whose analysis is our main objective, is the one where variables hold arithmetic values and data pointing to a read-only database.
%Our investigation comes with two main novelties with respect to the state of the art. On the one hand, we do not fix specific datatypes to which our results apply, but work under very mild and general assumptions: that the first-order theory underlying the each datatype is decidable, and enjoys the model-theoretic property of model completion, a weak form of quantifier elimination. On the other hand, while the analysis of such rich data-aware processes has so far mainly tackled safety properties, we provide for the first time foundational and practical results on their linear-time verification. 
%
Specifically, we provide four contributions. First, we devise a semi-decision procedure for linear-time verification of DMTs, which works for a very large class of datatypes obeying to mild model-theoretic assumptions. The procedure relies on a unique combination of automata-theoretic and cover computation techniques to respectively deal with linear-time properties and datatypes.
Second, we identify an abstract, semantic property that guarantees the existence of a faithful finite-state abstraction of the original system, and show that our method becomes a decision procedure in this case.
Third, we identify concrete, checkable classes of systems that satisfy this property, generalising several results in the literature. Finally, we present an implementation and a first experimental evaluation.

\end{abstract}

\newcommand{\ddmt}{DDSMT\xspace}

% 1.  DECIDABILITY OF SAS BEYOND SAFETY (LTLf MODEL CHECKING)
%     1a. ARITHMETIC SUPPORT (FRAGMENTS)
% 2.  GENERALIZATION OF FEEDBACK-FREE SYSTEMS WITH ARTIFACT VARIABLES
% 3.  GENERALIZATION OF THEOREM 5 OF BOJANCZYK BEYOND REACHABILITY

\section{Introduction}

In many application domains of AI, the evolution of dynamic systems is inextricably intertwined with the progression of some form of data. This calls for representing the states of the system with richer structures than propositional assignments, making the system intrinsically infinite-state. Notable examples of such \emph{data-aware processes} are 
(Situation Calculus) action theories~\cite{BaralG15,DeLP16},
lifted planning and planning over relational states/ontologies~\cite{FrancesG16a,CMPS16,BKKN22},
dynamic systems operating over databases~\cite{CGM13,DLV19,DHLV18,Gianola2023} and lightweight ontologies \cite{BCMD13,CGMM23}, and 
work processes in business process management (BPM)~\cite{Reichert12,CGGMR19-bpm}. In these settings, verification is especially important due to the data-process interplay, but highly challenging due to infinity of the state space.

When considering, within this large spectrum, those approaches that come with foundational results paired with effective algorithmic techniques and implementations, we observe that they either:
\begin{inparaenum}[\itshape (i)] 
\item operate on arbitrary datatypes while limiting verification to safety/reachability properties \cite{CGGMR20,GGMR23}, 
\item handle %full-fledged 
general
linear-time properties but only for specific datatypes, such as relational structures \cite{verifas}, or numerical variables~\cite{DD07,FMW22a},
\item restrict both the model and the verification formalism \cite{DFGM17}. 
\end{inparaenum}

In this paper, we tackle such limitations and study linear-time verification for so-called \emph{data-aware processes modulo theories} (\plainmydds{s}), from the foundational and practical point of view. Our verification machinery is very general: \plainmydds{s} capture dynamic systems manipulating variables that can store arbitrary types of data, ranging over infinite domains and equipped with domain-specific predicates, only imposing very mild assumptions on their underlying theory. 
In this respect, \plainmydds{s} subsume several models studied in the literature for which very little is known regarding decidability of linear-time verification. Among them, particularly relevant and investigated are those where variables can store numeric values subject to linear arithmetic operations, and/or objects extracted from a read-only relational database or lightweight ontology \cite{DD07,DDV12,BojanczykST13,CGGMR20,FMW22a,CGMM23}. Notably, analysis is in this case cast as a parameterized verification problem, where the property of interest needs to hold irrespectively of how the read-only component is instantiated. Since process executions typically have a finite (yet unbounded) length, we express properties using a data-aware extension of \LTLf~\cite{dGV13}.
However, data variables cause undecidability already for  reachability over very limited systems \cite{BojanczykST13}.

We show a simple example of our verification problem.

% commands for variable names
\newcommand{\xcst}{\mathit{c}}
\newcommand{\xvip}{\mathit{vip}}
\newcommand{\xtotal}{\mathit{t}}
\newcommand{\xbank}{\mathit{a}}
\newcommand{\xitm}{\mathit{p}}
\begin{example}
\label{exa:intro}
In a webshop, a customer logs in using a customer database, and chooses five products from a product database. It is
checked whether the customer is eligible for a 20\% discount, and, if so, the discount is applied to the price of the order. If the balance in the customer's account is sufficient,
the order is shipped; if not the process is restarted.
The system can be modeled as a guarded transition system:

\noindent
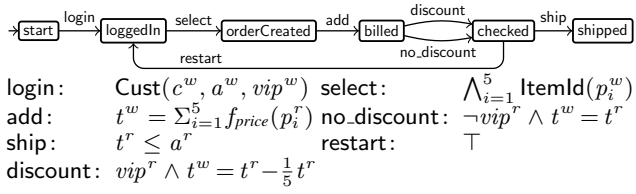
{
\begin{tikzpicture}[node distance=30mm]
\tikzstyle{state}=[draw, rectangle, rounded corners=1pt, inner sep=3pt, line width=.7pt, scale=.6]
\tikzstyle{action}=[scale=0.6]
\node[state] (1)  {$\m{start}$};
\node[state, right of=1, xshift=-9mm] (2) {$\m{loggedIn}$};
 \node[state, right of=2] (3) {$\m{orderCreated}$};
  \node[state, right of=3, xshift=-5mm] (4) {$\m{billed}$};
\node[state, right of=4, xshift=-3mm] (5) {$\m{checked}$};
% \node[state, right of=4, xshift=2mm] (5) {$\m{discountChecked}$};
\node[state, right of=5, xshift=-8mm] (6) {$\m{shipped}$};
\draw[edge] ($(1) + (-.4,0)$) -- (1);
\draw[edge] (1) to node[action, above=0mm] {$\m{login}$} (2);
\draw[edge] (2) to node[above=0mm,action]{$\m{select}$} (3);
\draw[edge] (3) to node[above=0mm,action]{$\m{add}$} (4);
\draw[edge,bend left =12] (4) to node[above=0mm,action]{$\m{discount}$} (5);
\draw[edge,bend right = 12] (4) to node[below=0mm,action]{$\m{no\_discount}$} (5);
\draw[edge, rounded corners] (5) -- ($(5) + (0,-.5)$) -| node[action, above=0mm, xshift=15mm] {$\m{restart}$} (2);
\draw[edge] (5) to node[above=0mm,action]{$\m{ship}$} (6);
\end{tikzpicture}\par
}
\noindent
\begin{footnotesize}
$\begin{array}{@{}l@{\:}l@{\,}l@{\:}l@{}}
\m{login}\colon &\m{Cust}(\xcst^w,\xbank^w,\xvip^w) &
\m{select}\colon &\bigwedge_{i=1}^5 \m{ItemId}(\xitm_i^w) \\
\m{add}\colon  &\xtotal^w=\Sigma_{i=1}^5 f_{\text{price}}(\xitm_i^r) &
\m{no\_discount}\colon  &\neg \xvip^r \land \xtotal^w\,{=}\, \xtotal^r \\
\m{ship}\colon  &\xtotal^r\leq \xbank^r &
\m{restart}\colon &\top \\
\m{discount}\colon  &\xvip^r \land \xtotal^w\,{=}\,\xtotal^r{-}\frac 1 5 \xtotal^r 
\end{array}$
\end{footnotesize} 
\\
The process uses variables $\xcst$, $\xbank$, $\xvip$ for cus\-tomer id, account balance, and eligibility for dis\-counts; $\xitm_i$ for product ids; and $\xtotal$ for the total cost. These variables are read (superscript $^r$) or written (superscript $^w$) as specified by the transition guards.
The DB relations $\m{Cust}(\mathit{CustomerId},$ $\mathit{Account},\mathit{IsVIP})$ and 
$\m{Items}(\mathit{ItmId}, \mathit{Price})$ hold customer and product data.
%
% In the latter, $\mathit{ItmId}$ is a primary key; we model this by a unary relation $\m{ItemId}(\mathit{ItmId})$ and a
% function symbol $f_{\text{price}}$ mapping the item id to a price
% (following the database model in \cite{CGGMR20}).
Relevant verification properties are for example that discounts are only applied to eligible customers, and that each order is eventually shipped.
\end{example}

%In this paper, we propose \plainmydds{s} as dynamic systems based on a general theory $\TT$, a concise yet expressive model that subsumes several existing formalisms~\cite{BojanczykST13,DDV12,CGGMR20,FMW22a}, and can e.g. formalize problems like in \exaref{intro}.
%
%

To address problems like the one of \exaref{intro}, we develop a semi-decision procedure that constructs a faithful, symbolic abstraction of the state space. While in purely numeric settings this idea has been explored by relying on quantifier elimination \cite{FMW22a}, the same approach cannot be lifted to theories formalizing data structures and relational databases, as quantifiers cannot be eliminated therein. We deal with this essential technical difficulty by combining automata-theoretic and cover computation \cite{CalvaneseGGMR21,CalvaneseGGMR22,GM08} %\todo{Ale: che citazione mettiamo?} 
techniques to respectively deal with the temporal and datatype dimensions.
%
%The basic idea resembles the approach used by~\citeauthor{FMW22a} (\citeyear{FMW22a}) for the purely numeric setting. However, this approach falls short for theories modelling data structures and databases, due to severe technical difficulties, related to the fact that in logical theories to model databases, quantifiers cannot be eliminated. 
%We resolve this problem by lifting all reasoning to the \emph{model completion} of a theory~\cite{CGGMR20}. 
%This allows us to do \LTLf model checking of data-aware processes that query read-only databases and operate on numeric variables, like \exaref{intro}.
%
% we unify and extend the two aforementioned recent lines of research. This is required as linear-time verification procedures were developed so far for processes over purely numeric variables only, exploiting quantifier elimination to construct faithful, logical abstractions of reachable states~\cite{FLM19,FMW22a}. However, this approach falls short for theories modelling data structures and databases, since such theories lack quantifier elimination. This can be tamed by lifting reasoning to the \emph{model completion} of a theory~\cite{CGGMR20}.  
% Many interesting theories satisfy this requirement, including equality with uninterpreted functions (EUF) and combinations of EUF with arithmetics.
%
More specifically, this paper makes the following four contributions:
\begin{inparaenum}[(1)]
%Sarah moved this sentence up
% \item
% We propose \plainmydds{s} as dynamic systems based on a general theory $\TT$, a concise yet expressive model that subsumes several existing formalisms~\cite{CGGMR20,BojanczykST13,DDV12,FMW22a}, and can e.g. formalize problems like in \exaref{intro}.
\item
We devise a model checking procedure for \plainmydds{s} w.r.t.~theories enjoying two mild theoretical assumptions: decidable
satisfiability of quantifier-free formulas, and the 
%existence of model completions 
computability of covers
\cite{CalvaneseGGMR21}. %\todo{Ale: cosa citiamo qua?}
\item
We propose the abstract criterion of \emph{\theproperty}, and show that for systems enjoying this criterion, our model checking procedure is a decision procedure.
\item
This  decidability criterion is shown to apply to several concrete classes of
systems singled out in the literature, where guards combine database queries with arithmetic.
% , where arithmetic constraints are restricted to certain shapes and the theory modeling database relations is either acyclic or locally finite.
% \todo{Questo mi sembra troppo specifico, ma non so bene come generalizzarlo. Forse potremmo dire for classes of arithmetic constraints and database schemas that have been singled out in the literature. Che ne dite?}
% \todo{S: done}
% and systems with the \emph{bounded lookback} property, a control-flow restriction. 
The property of \theproperty strictly generalizes several known decidability results from the literature (see below). %of related transition systems~\cite{CGGMR20,BojanczykST13,DDV12,FMW22a}.
\item We demonstrate the feasibility of our approach by an SMT-based implementation of the model checking technique, and an initial experimental evaluation.
\end{inparaenum}

In the sequel, after discussing related work, we introduce the \plainmydds model and the verification logic. We then
describe the model checking procedure and how to build automata to capture data-aware \LTLf properties. Next, we give a general decidability criterion and single out relevant, concrete decidable classes.
We close by discussing implementation and experiments.
Full proofs are in the appendix.

\paragraph{Related Work.}
% To the best of our knowledge, this is the first work proposing (semi-)decision procedures for linear-time verification in such a general setting, i.e., for transition systems operating over a decidable theory with model completion, with corresponding decidability results.
Notable approaches to the verification of dynamic systems operating over (read-only) relational databases \cite{BojanczykST13,CGGMR20} or lightweight ontologies \cite{CGMM23} focused on
safety properties: we generalize all such approaches, since we support full \LTLf verification.
\cite{BojanczykST13} is based on amalgamation, but it is known that the 
%existence of model completions 
computability of covers implies amalgamation \cite{CK}.
However, \cite{CGGMR20} also supports the richer setting of \emph{relational artifact systems}, where elements extracted from the database can be inserted into updatable relations (thus going beyond states containing variables). 
Several works deal with linear-time verification of systems operating over purely numeric data, with no support of other datatypes \cite{FMW22a,DD07,Demri06}.
A linear-time verification procedure was proposed for artifact systems with data dependencies and arithmetic \cite{DDV12}; this was shown to be a decision procedure for so-called \emph{feedback free} systems, which restrict how operations can be chained over time. 
This is the only decidability result combining DBs and arithmetic that we are aware of.
Our decidability criteria strengthen this result to the larger class of \emph{bounded lookback} systems. % \cite{FMW22a}.
A procedure for restricted linear-time verification of transition systems operating over databases was also presented by \citeauthor{DLV19} (\citeyear{DLV19}), but the verification language is not full LTL, systems need to be \emph{hierarchical}, and no arithmetic is supported.
Related to our work is also a tableau-based semi-decision procedure for satisfiability of \LTLf with general SMT constraints %using a tableau algorithm
\cite{GGG22,GGGW23}. %, but no decidability results are given there. % and they do not handle model checking.
%
% Finally, the use of model completion, also known as uniform interpolation, is not new for verification, but was so far employed only to check safety
% properties~\cite{CGGMR20}.
% In this work, the main ingredient to exploit to make the verification machinery fully implementable is given by quantifier elimination in model completion, which is equivalent to the operational counterpart provided by uniform interpolants. Notably, such interpolants can be efficiently computed for database theories via different methods, such as a constrained version of the Superposition Calculus  \cite{CalvaneseGGMR21}, tableaux-based procedures \cite{LMCS22}, and, in case of combined theories, by using general combination procedures via Beth Definability \cite{CGGMR22}. 
% Methods of this type are successfully leveraged in this paper to efficiently attack DMT verification. 

\section{Framework}
\label{sec:background}

%In this section we recall some background from the area of first-order logic, and fix our model for transition systems as well as the verification language.
We start with the necessary preliminaries.
% \paragraph{Logic.}
We consider a first-order multi-sorted \emph{signature} $\Sigma=\langle\SS, \PP, \FF, V, U\rangle$, where 
$\SS$ is a set of sorts;
$\PP$ is a set of predicate and $\FF$ a set of function symbols over $\SS$;
$V$ is a finite, non-empty set of \emph{data variables}; and
$U$ is a set of variables disjoint from $V$ that will be used for quantification; where all variables have a sort in $\SS$.
We assume that $\Sigma$ contains equality predicates for all sorts in $\SS$.
Then, $\Sigma$-terms $t$ are built in the usual way from $\FF$, $V$, and $U$.
An \emph{atom} is of the form $p(t_1,\dots, t_k)$ for $p\inn \PP$ and terms $t_1, \dots, t_k$; and a literal is an atom or its negation.
For a set of variables $Z$, a \emph{$\Sigma$-constraint} $c$ over $Z$ is of the form $\exists u_1, \dots, u_l.\phi$ such that $u_1, \dots, u_l\in U$, $\phi$ is a conjunction of $\Sigma$-literals, and all free variables in $c$ are in $Z$;
the set of all such constraints is denoted
$\CC_\Sigma(Z)$.
% A \emph{$\Sigma$-constraint} $c$ over variables $V$ is of the form $\exists u_1, \dots, u_l.\phi$ such that $u_1, \dots, u_l\in U$, $\phi$ is a conjunction of literals, and all free variables in $c$ are in $V$.
% The set of $\Sigma$-constraints over free variables $Z$ are denoted
% $\CC_\Sigma(Z)$, for any $Z$.
Moreover, \emph{$\Sigma$-formulas} $\phi$ have the form
$
\phi ::= p(t_1,\dots, t_k) \mid \phi \wedge \phi \mid \neg \phi \mid \exists u.\:\phi
$
where $p\in \PP$ and $u\in U$.
We use the usual shorthands $\forall$, $\vee$, and $\leftrightarrow$
for universal quantification, disjunction, and equivalence.
We call $\phi$ a \emph{state formula} if all its free variables are in $V$.
$\Sigma$-formulas without free variables are \emph{$\Sigma$-sentences}, and a
set of $\Sigma$-sentences is a 
\emph{$\Sigma$-theory} $\TT$. It is \emph{universal} if all its sentences have the form $\forall u_1, \dots, u_l.\, \phi$ with $\phi$ quantifier-free.

To define semantics, we use the standard notion of a \emph{$\Sigma$-structure} $M$,
which associates each sort $s\in \SS$ with a domain $s^M$, and each predicate $p\in \PP$ and function symbol $f\in \FF$ with a suitable interpretation $p^M$ and $f^M$. The equality predicates have the standard interpretation given by the identity relation.
The carrier of $M$, i.e., the union of all domains of sorts in $\SS$, is denoted by $|M|$.
% A \emph{$\Sigma$-theory} $\TT$ is a set of $\Sigma$-structures. \todo[inline]{Ale: do we use the SMT version of theories, i.e., the one with class of models? We can for sure, but then we should pay attention on how to give the definition of model completion}
%
% Given a $\Sigma$-structure $M$, we write $\Sigma^{|M|}$ for the signature $\Sigma$ extended with fresh constants for all elements in the carrier of $M$; then, $M$ can be naturally considered a $\Sigma^{|M|}$ structure by interpreting all additional constants as themselves.
%%%%%
A total function $\alpha\colon V \to |M|$ is called a \emph{state variable assignment}.
We also use $\alpha$ to denote a partial assignment 
$\alpha \colon V\cup U \to |M|$ with $V \subseteq Dom(\alpha)$.
We always assume that variables are mapped to an element of their respective domain. We write $\alpha[u \mapsto e]$ for the extended assignment such that
$\alpha[u \mapsto e](u)=e$ and $\alpha[u \mapsto e](x)=\alpha(x)$ for all $x\neq u$.

Given $M$ and $\alpha$, the \emph{evaluation} $\eval{t}$ of a term $t$ is defined as $\eval{v} = \alpha(v)$ if $v\inn V \cup U$, and $\eval{f(t_1,\dots, t_k)} = f^M(\eval{t_1}, \dots, \eval{t_k})$.
Whether a $\Sigma$-formula $\phi$ is \emph{satisfied} by $M$ and $\alpha$, denoted
$\alpha \models_M \phi$, is defined as follows:\\
\begin{tabular}{@{}l@{\:}l@{}}
$\alpha \models_M p(t_1,\dots, t_k)$
 & if $p^M(\eval{t_1}, \dots, \eval{t_k})$ holds\\
$\alpha \models_M \phi_1\wedge\phi_2$
 & if $\alpha  \models_M \phi_1$ and $\alpha  \models_M \phi_2$\\
$\alpha  \models_M \neg \phi$
 & if $\alpha  \not\models_M \phi$\\
$\alpha  \models_M \exists u.\:\phi$
 & if $\exists$ $e\in |M|$ s.t. $\alpha[u \mapsto e] \models_M \phi$
\end{tabular}\\
Note that $\alpha \models_M \phi$ is always defined if $\phi$ is a state formula and $\alpha$ a state variable assignment. 
If $\phi$ is a $\Sigma$-sentence, we write
$\models_M \phi$ for $\emptyset \models_M \phi$.
% A function $\alpha\colon V \to |M|$ is called a \emph{state variable assignment}, while a function $\gamma\colon U\to |M|$ is an \emph{environment},
% where we assume in both cases that all variables are mapped to an element of their respective domain. We write $\gamma[u \mapsto e]$ for the environment $\gamma$ extended with a binding from $u$ to $e$.
% 
% Given $M$, $\alpha$, and $\gamma$, the \emph{evaluation} $\eval{t}$ of a term $t$ is defined as $\eval{v} = \alpha(v)$ if $v\inn V$, $\eval{u} = \gamma(u)$ if $u\inn U$, and $\eval{f(t_1,\dots, t_k)} = f^M(\eval{t_1}, \dots, \eval{t_k})$.
% Whether a $\Sigma$-formula $\phi$ is \emph{satisfied} by $M$, $\alpha$, and $\gamma$, denoted
% $\alpha,\gamma \models_M \phi$, is defined as follows:\\
% \begin{tabular}{@{}l@{\:}l@{}}
% $\alpha,\gamma \models_M p(t_1,\dots, t_k)$
%  & if $p^M(\eval{t_1}, \dots, \eval{t_k})$ holds\\
% $\alpha,\gamma \models_M \phi_1\wedge\phi_2$
%  & if $\alpha,\gamma \models_M \phi_1$ and $\alpha,\gamma \models_M \phi_2$\\
% $\alpha,\gamma \models_M \neg \phi$
%  & if $\alpha,\gamma \not\models_M \phi$\\
% $\alpha,\gamma \models_M \exists u.\:\phi$
%  & if $\exists$ $e\in |M|$ s.t. $\alpha,\gamma[u \mapsto e] \models_M \phi$
% \end{tabular}\\
% For a state formula $\phi$, we also write $\alpha \models_M \phi$ to express
% $\alpha,\emptyset \models_M \phi$, and if $\phi$ is a $\Sigma$-sentence, we write
% $\models_M \phi$ for $\emptyset \models_M \phi$.
%%%%%
We also write $M \in \TT$ to denote that $M$ is a model of $\TT$, i.e.,
$\models_M \phi$ holds for all %sentences 
$\phi$ in $\TT$.
A state formula $\phi$ is \emph{$\TT$-satisfiable} if there is some $M\in \TT$ and state variable assignment $\alpha\colon V \to |M|$ such that $\alpha \models_M \phi$.
% and $\phi$ is \emph{$\TT$-valid} if for all $M\in \TT$ and all $\alpha$ it holds that $ \alpha \models_M \phi$.
%
Moreover, two state formulas $\phi_1$ and $\phi_2$ are \emph{$\TT$-equivalent}, denoted
$\phi_1 \equiv_\TT \phi_2$, if $\neg(\phi_1 \leftrightarrow \phi_2)$ is not $\TT$-satisfiable.

A $\Sigma$-theory $\TT$ has \emph{quantifier elimination} (QE) if for every $\Sigma$-formula $\phi$
 there is a quantifier-free formula $\phi'$ that is $\TT$-equivalent to $\phi$.
%As QE is a strong requirement,
%we make use of the weaker notion of a $\Sigma$-theory $\TT$ 
%having \emph{uniform interpolation}, 
%that allows to \emph{compute covers}
%or, equivalently~\cite{CGGMR20}, having a \emph{model completion} \cite{CK}: a universal $\Sigma$-theory $\TT$ has a model completion iff there exists a $\Sigma$-theory $\TT^*$ such that
%$(i)$ every $\Sigma$-constraint satisfiable in a model of $\TT$ is satisfiable in a model of $\TT^*$, and
%$(ii)$ $\TT^*$ has QE~\cite{Ghilardi04}.
As QE is a strong requirement,
we make use of the weaker notion of covers: a $\TT$-cover of an existential formula $\psi$ is the strongest formula $\TT$-implied by $\psi$.
Formally,  computing covers in $\TT$ is equivalent to having a \emph{model completion} \cite{CalvaneseGGMR21}: a universal $\Sigma$-theory $\TT$ has a model completion iff there exists a $\Sigma$-theory $\TT^*$ such that
$(i)$ every $\Sigma$-constraint satisfiable in a model of $\TT$ is satisfiable in a model of $\TT^*$, and
$(ii)$ $\TT^*$ has QE~\cite{Ghilardi04,CK}. Throughout the paper, in the formal proofs we make use of model completions, since they are easier to handle when adopting, as we do, a model-theoretic approach.

We will sometimes refer to common SMT theories~\cite{BarrettSST21}: the theory of equality and uninterpreted functions  (EUF) for a given $\Sigma$, and linear arithmetic over rationals (LRA), integers (LIA), or both (LIRA). While the arithmetic theories have QE~\cite{Presburger29}, EUF admits model completion and so do certain \emph{tame} combinations of LIRA and EUF~\cite{CGGMR19,CalvaneseGGMR22}.

% For instance, to model the process in \exaref{intro} we can consider a signature $\Sigma$ with uninterpreted sorts $\mathit{CustomerId}$ and $\mathit{ItmId}$, additional sort $\mathit{Status}$ to model the control states, and arithmetic sort $rat$; the set of function symbols $\mathcal F$ consisting of $f_{\mathit{price}}\colon \mathit{ItmId} \to rat$, constants $\undef_s$ for all uninterpreted sorts to indicate undefinedness,
% and constants for all stati like $\m{start}$, $\m{loggedIn}$, $\m{orderCreated}$, etc. together with arithmetic operations. The set of predicates $\mathcal P$ consists of $\m{Cust}$ and $\m{ItemId}$, arithmetic predicates, and equality. % on all sorts.

\paragraph{Data-Aware Processes Modulo Theories.}
%We next define our model of transition systems.
Let $\Sigma=\langle\SS, \PP, \FF, V, U\rangle$ be a signature.
 %enrich data-aware dynamic systems \cite{LFM20} with $\Sigma$-constraints as follows.
For each data variable $v\in V$, 
let $v^r$ and $v^w$ be two annotated variables of the same sort, and set $V^r = \{v^r \mid v\in V\}$ and $V^w = \{v^w \mid v\in V\}$. These copies of $V$ are called the \emph{read} and \emph{write} variables; they will denote the variable values before and after a transition, respectively. 
Moreover, $\vec V$ denotes a vector that sorts the variables $V$ in an arbitrary but fixed order. %, and similarly for $\vec V^r$ and $\vec V^w$.

\begin{definition}
\label{def:dds}
A \emph{data-aware process modulo theories} over $\Sigma$ (\mydds for short) is
a labelled transition system $\BB = \dmttuple$, where:
\begin{inparaenum}[\itshape (i)]
\item $\Sigma$ is a signature, % and $\TT$ is a $\Sigma$-theory (called \emph{underlying theory});
\item $V$ is the finite, nonempty set of data variables in $\Sigma$;
\item the transition formulae $T(V^r,V^w)$ %, called \emph{symbolic transition template}, 
are a set of constraints in $\CC_\Sigma(V^r\cup V^w)$;
% \item $F(V)$ is a set of $\Sigma$-constraints in $\Sigma(V)$, called \emph{set of final states};
\item $I\colon V\to \FF_0$, called initial function, 
initializes variables,
where $\FF_0$ is the set of $\Sigma$-constants.
\end{inparaenum}
\end{definition}

\newcommand{\setx}{\m{xset}}
\newcommand{\sety}{\m{yset}}
\begin{example}
\label{exa:simple}
As a running example, we consider a simple \plainmydds $\BB = \dmttuple$ where $\Sigma$ combines LRA with EUF using uninterpreted sorts $status$ and $elem$, uninterpreted predicates $\m R \subseteq rat \times elem$, $\m P\subseteq elem$ and constants $\m a$, $\m b$, $\m o_1$, $\m o_2$. The set of variables $V$ consists of $x$ (sort $rat$) and $y$ (sort $elem$), and  a variable $s$ of sort $status$ that takes values $\m o_1$ and $\m o_2$. We set $I(s){=}\m o_1$, $I(x){=}0$ and $I(y){=}\m a$. The transitions $T = \{\m{setx}, \m{sety}\}$ are as follows:

{\centering
\begin{footnotesize}
\begin{tabular}{r@{\,=\,}l}
$\setx$ &
$\left (s^r \eqn \m o_1 \wedge s^w \eqn \m o_2 \wedge x^w {>}x^r \wedge \m R(x^w,y^r) \right)$ \\
$\sety $ &
$\left (s^r \eqn \m o_2 \wedge s^w \eqn \m o_1 \wedge \m P(y^w)\right)$
\end{tabular}
\end{footnotesize}
\qquad
\raisebox{-4mm}{
\begin{tikzpicture}[node distance=50mm,>=stealth']
\node[state] (1) {$\mathsf o_1$};
\node[state, right of=1] (2) {$\mathsf o_2$};
\draw[edge] ($(1) + (-.4,0)$) -- (1);
\draw[edge, rounded corners] (1) -- node[above,action]{$\setx\colon [x^w {>}x^r \wedge \m R(x^w,y^r)]$} (2);
\draw[edge, rounded corners] (2) -- ($(2) - (0,.4)$)
  -- node[below,action]{$\sety\colon [\m P(y^w)]$}  ($(1) - (0,.4)$)  -- (1);
\end{tikzpicture}}\par
}
\noindent
The transition system is a visualization of $\BB$ where the status $s$ is interpreted as a control state, to help readability.
Transitions simultaneously express conditions on read variables (superscript $^r$), and updates on the written ones (super\-script~$^w$):  e.g., $\setx$ 
requires the current control state ($s^r$) to be $\m o_1$, fixes the next control state ($s^w$) to $\m o_2$, and
nondeterministically sets $x$ to a new value ($x^w$) that is greater than its current value ($x^r$), and in relation $\m R$ with $y^r$.
\end{example}
% 
% \begin{example}
Similarly, also \exaref{intro} can be formalized as a \mydds.
% where $\Sigma$ contains the symbols introduced there.
%
% The set $V$ consists of the
% variables $\xcst$ of sort $\mathit{customerId}$, $\xitm_i$ of sort $\mathit{itmId}$, $\xtotal$ and $\xbank$ of sort $\mathit{rat}$, the boolean flag $\xvip$ can be modeled as sort $\mathit{rat}$ too, and the control state can be modeled by an additional variable $s$ of uninterpreted sort $\mathit{status}$ that takes values $\m{start}$, $\m{loggedIn}$, etc.
% As initial function we can e.g. choose $I(\xbank)=I(\xtotal)=0$, $I(s)=\m{start}$,
% and $I(x) = \undef$ for all other variables $x$.
% Again, transitions read and update variables:
% e.g., $\m{sum\_up}$ reads the values of all $\xitm_i$, and writes to $\xtotal$ the sum of their prices; while $\m{login}$ nondeterministically  picks values for $\xcst$, $\xbank$, and $\xvip$ that are in the relation $\m{Cust}$.
% \end{example}

% Note that by the definition of constraints, \defref{dds} excludes guards with disjunctions, but these can be modeled by multiple transitions between the same states.
%
We next define the semantics for {\mydds}s.
For a $\Sigma$-theory $\TT$ and a model $M\in \TT$,
a \emph{state} of $\BB$ is an assignment $\alpha\colon V \to |M|$.
A \emph{guard assignment} $\beta$ is a function $\beta\colon V^r \cup V^w \mapsto |M|$.
As defined next, a transition $t$ can transform a state $\alpha$ into a new state $\alpha'$, updating the variable values in agreement with $t$,
while variables that are not explicitly written keep their previous value as per $\alpha$.
% In the new $\alpha'$, variables that are not written keep their previous value as per $\alpha$, while written variables are updated according to the transition formula.

\begin{definition}
% A \mydds $\BB$ admits a \emph{$\TT$-step} $\alpha$ to 
% $\alpha'$ via action $a$ w.r.t. model $M\in \TT$,
% denoted $(b, \alpha) \goto{a}_M (b', \alpha')$,
% if $b \goto{a} b'$ and $\beta \models_M \guard(a)$ for
% the guard assignment $\beta$ s.t.
% $\beta(v^r) = \alpha(v)$ and
% $\beta(v^w) = \alpha'(v)$ for all $v \in V$.
A \mydds $\BB = \dmttuple$ admits a \emph{$\TT$-step} from state $\alpha$ to 
$\alpha'$ via transition $t \in T$ w.r.t. a model $M\in \TT$,
denoted $\alpha \goto{t}_M \alpha'$,
if there is some guard assignment $\beta$ s.t. $\beta \models_M t$,
$\beta(v^r) = \alpha(v)$ and
$\beta(v^w) = \alpha'(v)$ for all $v \in V$.
\end{definition}

\noindent
% A \emph{$\TT$-run} of $\BB$ is a pair $(M,\rho)$ of a model $M\in \TT$
% and a sequence of steps of the form
% $\rho\colon(b_0, \alpha_0) 
% \goto{a_1}_M (b_1, \alpha_1)
% \goto{a_2}_M \dots
% \goto{a_n}_M (b_n, \alpha_n)$, and $n$ is its length.
A \emph{$\TT$-run} of $\BB$ is a pair $(M,\rho)$ of a model $M\in \TT$
and a sequence of steps of the form
$\rho\colon \alpha_0
\goto{t_1}_M  \alpha_1
\goto{t_2}_M \dots
\goto{t_n}_M \alpha_n$ where $\alpha_0(v) = I(v)^M$ for all $v\in V$.
Note that given $M$, the initial assignment $\alpha_0$ of a run is uniquely determined by the 
initializer $I$ of $\BB$.
% The transition sequence 
% $\smash{\sigma\colon b_0
% \goto{a_1} b_1
% \goto{a_2} \dots \goto{a_n} b_n}$ with the same states and actions is called the \emph{abstraction} of $\rho$, and is denoted by $\sigma(\rho)$.
The transition sequence 
$\langle t_1, \dots, t_n\rangle$ is called the \emph{abstraction} of $\rho$, and is denoted by $\sigma(\rho)$.
% We also write $\rho_i$ to refer to $(b_i, \alpha_i)$.
% and denote the length $n$ of $\rho$ by $|\rho|$.
If clear from the context, we omit the model $M$
in the notation $\goto{t}_M$; and sometimes we refer to a \emph{run}, leaving the theory implicit. 
% 
% The prefix of $\sigma$ of $i$ steps is denoted $\sigma|_i$.

For \exasref{intro}{simple}, a natural choice for the theory $\TT$ is
that of arithmetic over $\mathbb Q$ together with equality and uninterpreted functions (i.e., the union of LRA and EUF), plus 
axioms that fix entries in the database.
E.g. for \exaref{simple}, if $M$ satisfies $\m R(4,\m a)$ and $\m R(8,\m b)$, then $(M,\rho)$ is a run, for
\begin{footnotesize}
\begin{equation*}
\label{eq:exrun}
\rho \colon
\left\{\begin{array}{@{}r@{\,=\,}l@{}}s&\m o_1\\[-.3ex]x&0\\[-.3ex]y&\m a\end{array}\right\} \goto{\setx}
\left\{\begin{array}{@{}r@{\,=\,}l@{}}s&\m o_2\\[-.3ex]x&4\\[-.3ex]y&\m a\end{array}\right\} \goto{\sety}
\left\{\begin{array}{@{}r@{\,=\,}l@{}}s&\m o_1\\[-.3ex]x&4\\[-.3ex]y&\m b\end{array}\right\} \goto{\setx}
\left\{\begin{array}{@{}r@{\,=\,}l@{}}s&\m o_2\\[-.3ex]x&8\\[-.3ex]y&\m b\end{array}\right\}
\end{equation*}
\end{footnotesize}
% 
% For instance, let $\TT$ contain the axioms $\mathit{Cust}(\m{joe}, 642, 1)$ and 
% \todo{update}
% $\mathit{Item}(1, 10)$, $\mathit{Item}(2, 20)$, $\mathit{Item}(3, 30)$.
% Then the following is a run for any model $M$ of $\TT$,
% where every line is a configuration and $\undef$ abbreviates undefinedness constants:\\
% \tikz{
% \node[scale=.9, inner sep=0pt]{
% \begin{tabular}{c|c|c|c|c|c|c|c|c|c|c}
% state & $x_{\mathit{status}}$ & $x_\mathit{cst}$ & $x_{\mathit{bank}}$ & $x_\mathit{vip}$ & $x_{\mathit{total}}$ & $x_{\mathit{itm},1}$ & $x_{\mathit{itm},2}$ & $x_{\mathit{itm},3}$& $x_{\mathit{itm},4}$ & $x_{\mathit{itm},5}$ \\
% \hline
% \text{\circled{1}} & \# & \# & 0 & 0 & 0 & \# & \# & \# & \# & \#\\
% \text{\circled{2}} & CustomerLoggedIn & $\m{joe}$ & 642 & 1 & \# & \# & \# & \# & \# & \#\\
% \text{\circled{3}} & OrderCreated & $\m{joe}$ & 642 & 1 & \# & 1 & 1 & 2 & 3 & 2\\
% \text{\circled{4}} & ShopCartUpdated & $\m{joe}$ & 642 & 1 & 90 & 1 & 1 & 2 & 3 & 2\\
% \text{\circled{5}} & DiscountAssigned & $\m{joe}$ & 642 & 1 & 72 & 1 & 1 & 2 & 3 & 2\\
% \text{\circled{6}} & OrderShipped & $\m{joe}$ & 642 & 1 & 72 & 1 & 1 & 2 & 3 & 2
% \end{tabular}
% };
% \foreach \i/\t in {1/login,2/select,3/sum, 4/discount, 5/ship}{
% 	\draw[bend right=50, ->] (-5.5,1.3 - .42 * \i) to node[anchor=east, scale=.8] {$\m{\t}$}(-5.5,.98 - .42 * \i);
% }
% }

\paragraph{Verification language.}
\label{sec:lang}
%We next fix the language to express verification properties.
From now on, we assume that $V$ is the set of data variables of a given \plainmydds $\BB$, and all variables in verification properties are in $V$.
Formally, our verification language $\LL_\Sigma$ consists of all properties $\psi$ defined as follows:

{\centering
$\psi := c \mid %b \mid 
\psi {\wedge} \psi \mid  \psi {\vee} \psi \mid 
% \langle a\rangle \psi \mid 
\X \psi \mid 
% \F \psi \mid 
\G \psi \mid \psi \U \psi$\par
}

\noindent
% Sarah removed: syntactic top, bot is not possible anyway
where $c$ is a constraint in $\CC_\Sigma(V)$. % \setminus \{\top,\bot\}$.
Note that $\LL_\Sigma$ does not support negation, but can express formulae in negation normal form. For convenience, we abbreviate $\top := (v=v)$ for some $v\in V$, and $\F \psi: = \top \U \psi$.
% If $\BB$ is clear from the context, we simply write $\LL$ instead of $\LL_\Sigma$.
We adopt LTL semantics over finite traces (LTL$_f$)~\cite{dGV13}:

\begin{definition}
\label{def:witness}
A run $(M,\rho)$ of $\BB$
for $\rho\colon \alpha_0
\goto{t_1}_M  \alpha_1
\goto{t_2}_M \dots
\goto{t_n}_M \alpha_n$
\emph{satisfies} $\psi \in \LL_\Sigma$, denoted 
$\rho \models_M \psi$, iff $\rho \models_M^0 \psi$ holds, 
where for all $i$, $0 \leq i \leq n$:

\noindent
\begin{tabular}{@{}l@{\:}l@{}}
% $\rho \models_M^i c$  & iff $\rho_i = (b,\alpha)$ for some $b$
%  and $\alpha$ such that $\alpha\models_M c$ \\
$\rho \models_M^i c$  & iff $\alpha_i\models_M c$ \\
% $\rho \models_M^i b$ & iff $\rho_i = (b,\alpha)$ for   
%  some $\alpha$\\
$\rho \models_M^i \psi_1 \wedge \psi_2$ & iff 
$\rho \models_M^i \psi_1$ and $\rho \models_M^i \psi_2$ \\
$\rho \models_M^i \psi_1 \vee \psi_2$ & iff 
$\rho \models_M^i \psi_1$ or $\rho \models_M^i \psi_2$\\
$\rho \models_M^i \X\psi$ &iff $i<n$ and 
 $\rho \models_M^{i{+}1} \psi$\\
$\rho \models_M^i \G\psi$ & iff 
$\rho \models_M^i \psi$ and ($i=n$ or
$\rho\models_M^{i{+}1} \G\psi$)\\
$\rho \models_M^i \psi_1 \U \psi_2$ & iff 
$\rho \models_M^i \psi_2$, or ($i\,{<}\,n$ and both \\
 & $\rho \models_M^i \psi_1$ and $\rho\models_M^{i+1} \psi_1 \U \psi_2$)
\end{tabular}
\end{definition}

% We use $\LL$ to express properties over the finite traces of a \mydds $\BB$. 
% To that end,
\noindent
A $\TT$-run $(M,\rho)$ is a \emph{$\TT$-witness} for $\psi \in \LL_\Sigma$ if 
% $\rho$ ends in a final state of $\BB$ and 
$\rho \models_M \psi$.
The problem addressed in this paper is the following:

\begin{definition}[$\TT$-verification task]
\label{def:verification:task}
Given a \mydds $\BB$, a $\Sigma$-theory $\TT$, and $\psi \in \LL_\Sigma$,
check whether there exists a $\TT$-witness for $\psi$ in $\BB$.
\end{definition}

% \noindent
% One can thus model check a $\Sigma$-DDS $\BB$ for $\psi$ by searching a witness for $\neg \psi$, i.e., a counterexample.

% \begin{example}
% \label{exa:order}
The \plainmydds in \exaref{simple} has a witness for $(x \geq 0) \U (s\eqn \m o_2 \wedge x\eqn4)$, e.g. the run $(M, \rho)$ above. 
The \plainmydds in 
\exaref{intro} has no witness for $\F (s{=}\m{shipped} \wedge \neg \xvip \wedge \xbank < \sum_{i=1}^5 f_{\mathit{price}}(\xitm_i))$, so  an order is not shipped if the account balance is insufficient.
% \end{example}

\begin{remark} 
\defref{verification:task} generalizes
related verification tasks for 
% transition systems studied in the tradition of 
data-aware processes, notably
(1) verification of safety properties in
Simple Artifact Systems (SASs) over relational databases with key dependencies  \cite{CGGMR20}, and (2) linear-time model checking of DDSs with arithmetics  \cite{FMW22a}. 
SASs can be represented as DMTs by taking the theory EUF with constants, unary functions, and arbitrary relations, possibly together with a read-only database.
% , cf. \cite{CGGMR20}. 
\emph{Artifact variables} of SASs can be represented as data variables in $V$, and safety properties can be expressed in $\LL_\Sigma$.
DDSs can be formalized as DMTs by taking the theory LIRA: The initial assignment of variables can be encoded in $I$, one additional designated variable can be used to model DDS control states, and the transition formulae $T$ in DMTs can represent guarded actions.
 DDSs have final states, but any verification property can be extended to require that such a state is reached.
%
% \defref{verification:task} generalizes
% expressive transition systems well studied in the tradition of data-aware processes:
% related verification tasks for transition systems studied in the tradition of data-aware processes, notably
% (1) verification of safety properties in
% Simple Artifact Systems (SASs) over relational databases with key dependencies  \cite{CGGMR20}, and (2) linear-time model checking of DDSs with arithmetics  \cite{FMW22a}. 
% %
% SASs can be represented as DMTs, where the \emph{artifact variables} in SASs are encoded via data variables in $V$.
% Checking safety properties can be expressed as our verification task for a suitable $\psi$ with respect to the theory of EUF with constants, unary functions, and arbitrary relations, possibly together with a background read-only database, cf. \cite{CGGMR20}.
% DDSs can be formalized as DMTs, by encoding via $I$ the initial assignment of the DDS, using one additional designated variable for representing DDS control states, and the transition formulae $T$ in DMTs to represent guarded actions between states in DDSs.
%  DDSs have final states $F$, but any verification property can be extended to require that a state in $F$ is reached.
% The verification task of \cite{FMW22a} is the captured by
% setting $\TT$ to LIRA in \defref{verification:task}.
\end{remark}

LTL$_f$ verification is undecidable for dynamic systems over numeric variables, as reachability in 2-counter machines can be encoded~\cite{FMW22a}, so
the verification task considered here is undecidable, too.

\section{Model Checking}
\label{sec:mc}
% \section{Automata for Properties}
% \label{sec:automata}

From now on, we assume a fixed $\Sigma$-theory $\TT$ that satisfies:

\begin{assumptions}
\label{ass:theory}
\begin{inparaenum}[(a)]
\item Satisfiability of quantifier-free $\TT$-formulas is decidable; 
\item  $\TT$ either has QE, or is universal and has a model completion $\TT^*$.
\end{inparaenum}
\end{assumptions}
% model completion does not imply decidable sat: nonlinear (polynomial) real arithmetic

\begin{remark}
\label{rem:TT*}
As constraints are existential formulas, Assumption (b) implies that
$\TT$-satisfiability and $\TT^*$-satisfiability
of $\Sigma$-constraints coincide~\cite{CGGMR20}.
\end{remark}

\paragraph{Automata construction.}
It is well-known that \LTLf properties can be captured by non-deterministic automata (NFA)~\cite{dGV13,GMM14}. Our model checking procedure relies on such an NFA for a given verification property $\psi\in \LL_\Sigma$.
Its construction differs from earlier work only in that propositional/numeric atoms are replaced by constraints.
We sketch here the general approach, while full details are in the appendix.
For a \mydds $\BB = \dmttuple$ and $\psi\in \LL_\Sigma$, let $C$ be the set of all constraints in $\psi$.
% The alphabet of $\NFA$ is $\smash{\Theta=2^{C}}$, i.e., a symbol is a subset of constraints in $C$.
We then build an NFA 
$\NFA\,{=}\,(Q, \Theta, \varrho, q_0, \{q_f, q_e\})$ where:
\begin{inparaenum}[\itshape (i)]
\item the set $Q$ of states consists of properties $\phi \in \LL_\Sigma\cup \{\top,\bot\}$,
\item the alphabet is $\smash{\Theta:=2^{C}}$, i.e., a symbol is a subset of constraints in $C$,
\item the initial state is
$q_0\,{:=}\,\inquotes{\psi}$,
\item $q_f\,{:=}\,\inquotes{\top}$ is a final state,  
\item $q_e$ is an additional, designated final state, %\todo{Marco: what does it represent???} 
\item $\varrho \subseteq Q \times \Theta \times Q$ is a set of transitions.
\end{inparaenum}

To state correctness, let a word $\langle\varsigma_0,\dots,\varsigma_n\rangle\in \Theta^*$ be \emph{consistent} with a run $(M,\rho)$ for
$\smash{\rho\colon \alpha_0 
\goto{t_1} \alpha_1
\goto{t_2} \dots \goto{t_n} \alpha_n}$ if
$\alpha_i \models_M \bigwedge \varsigma_i$
for all $0\,{\leq}\,i\,{\leq}\,n$.
The main result about $\NFA$ is:

\begin{restatable}{proposition}{propositionNFA}
\label{prop:nfa}
% \begin{compactenum}[(1)]
% \item If $\rho \models_M \psi$ then $\NFA$ accepts a word $w$ consistent with $(M,\rho)$.
% \item If $\NFA$ accepts a word $w$ that is consistent with $(M,\rho)$ then $\rho \models_M \psi$. 
% \end{compactenum}
$\NFA$ accepts a word $w$ that is consistent with $(M,\rho)$ iff $\rho \models_M \psi$. 
\end{restatable}

\noindent
% The proof is similar as \cite[Lem. 4.4]{FMW22a}, except that all satisfiability checks done in the $\owedge$ and $\ovee$ operators use $\TT$-satisfiability, as defined above, and the fact that a DMT has no control states. 
All details can be found in the appendix.
We show here instead the construction for an example.

\begin{example}
\label{exa:nfa}
The NFA $\NFA$ for $\psi := (x \geq 0) \U (s\eqn \m o_2 \wedge x\eqn4)$ is as shown below.
Note how the initial state $\psi$ has a loop labeled $x\geq 0$ that can be taken until both $s\eqn \m o_2$ and $x\eqn4$ hold; at this point the final state $\top$ can be reached. The construction is detailed in \exaref{nfa2}.

{\centering
\tikz[node distance=28mm]{
\tikzstyle{formula}=[scale=.75, rectangle, rounded corners=2pt, inner sep=2pt, draw]
\tikzstyle{edge}=[scale=.7]
\node[formula] (psi) {${\psi}$};
% \node[formula, left of=psi] (bot) {${\bot}$};
\node[formula, right of=psi, double] (top) {${\top}$};
\draw[->] ($(psi) + (-.3,.3)$) to (psi);
\draw[->] (psi) to 
 node[edge, above, yshift=0pt]{$\{x\eqn 4, s\eqn \m o_2\}$}
 (top);
% \draw[->] (psi) to node[edge, above, yshift=0pt]{$\emptyset$} (bot);
% \draw[->] (psi) to (bot);
\draw[->] (top) to[loop right, looseness=7] node[edge, right]{$\emptyset$} (top);
% \draw[->] (bot) to[loop left, looseness=6] node[edge, left]{$\emptyset$} (bot);
\draw[->] (psi) to[loop left, looseness=7] node[edge, left]{$\{x\,{\geq}\,0\}$}(psi);
}\par

}
\end{example}

\paragraph{Verification procedure.}
We next develop our model checking technique. To this end, we assume a $\Sigma$-theory $\TT$, 
a \mydds $\BB = \dmttuple$, and $\psi\in\LL_\Sigma$ with associated NFA $\NFA$ and alphabet $\Theta$ as above.
We first aim to show the crucial fact that witnesses exist in $\TT$ iff they do in $\TT^*$:

\begin{theorem}
\label{thm:T:witness}
$\BB$ admits a $\TT$-witness for $\psi \in \LL_\Sigma$ iff
it admits a $\TT^*$-witness.
\end{theorem}

To prove \thmref{T:witness}, we introduce notions that will be used throughout the paper:
Let $\writ(t) := \{x \mid x^w\in V^w\text{ occurs in }t\}$ be the set of variables written by a transition $t\in T$. We write $\trans{t}$ for the \emph{extended transition formula}
$\trans{t}:= t \wedge \bigwedge_{v\not\in \writ(t)} v^{w}\,{=}\,v^{r}$, which
simply expresses that the transition formula $t$ must hold and the values of all variables that are not explicitly written are propagated. 
Below, we consider disjoint variable sets $V_0, V_1, V_2,\dots$ that are indexed copies of $V$, i.e.
$V_i := \{v_i \mid v\in V\}$, and assume that $\vec V_0, \vec V_1, \vec V_2,\dots$
are ordered in the same way as $\vec V$.
% We call $\vec V_i$ a \emph{fresh copy} of $\vec V$ if
% for $\vec V = \langle v_1, \dots, v_k\rangle$, vector $\vec X = \langle x_1, \dots, x_k\rangle$ such that $x_i$ has the same sort as $v_i$, and $X$ and $V$ are disjoint.
For a formula $\phi$ with free variables $V$, we write $\phi(\vec V_i)$ for the formula obtained by renaming $\vec V$ to $\vec V_i$; and for $t$ with free variables $V^r\cup V^w$,
the formula $t(\vec V_i, \vec V_j)$ replaces $\vec V^r$ by $\vec V_i$ and $\vec V^w$ by $\vec V_j$.
Moreover, $\phiinit := \bigwedge_{v\in V} v_0{=}I(v)$ is a formula that encodes the initial state.

We next define a formula that combines all constraints accumulated in a transition sequence, together with a word $w$ that expresses additional constraints relevant for verification.
For readability, we use $\varsigma \in \Theta$ as a formula to mean $\bigwedge \varsigma$.

\begin{definition}
\label{def:H}
For a transition sequence $\sigma {=} \langle t_1, \dots, t_n\rangle$ of $\BB$ and  $w=\langle\varsigma_0,\dots, \varsigma_n\rangle\in \Theta^*$, the formula %\emph{history expansion} 
$H(\sigma,w)$ is %defined as

\noindent
$\phiinit \wedge 
\varsigma_0(\vec V_0) \wedge
\trans{t_1}(\vec V_0, \vec V_1) \wedge \varsigma_1(\vec V_1) \wedge
\dots
\wedge \trans{t_n}(\vec V_{n-1}, \vec V_n) \wedge \varsigma_n(\vec V_n)
$
\end{definition}
E.g., for the transition sequence $\sigma = \langle \setx\rangle$ of the \plainmydds in \exaref{simple} and
$w = \langle \{x\,{\geq}\,0\}, \{x{=}4,s{=}\m o_2\}\rangle$,
$H(\sigma, w) = (s_0{=}\m o_1 \wedge x_0{=}0 \wedge y_0{=}\m a) \wedge (x_0{\geq}0) \wedge
(s_0{=}\m o_1 \wedge s_1{=}\m o_2 \wedge x_1{>}x_0 \wedge \m R(x_1, y_0) \wedge y_1{=}y_0 \wedge x_1{=}4))$.
In \lemref{abstraction} below we prove that satisfying assignments for formulas $H(\sigma,w)$ are witnesses for $\psi$, if $w$ is accepted by $\NFA$. 
To state this, we need some notation to link assignments with runs:

\begin{definition}
Let $\sigma = \langle t_1, \dots, t_n\rangle$ be a transition sequence, $w\in \Theta^*$, and $M$ be a $\Sigma$-structure.
\begin{compactitem}
\item
If $(M,\rho)$ is a run of $\BB$ for
$\rho \colon \alpha_0
\goto{t_1} \alpha_1
\goto{t_2} \dots
\goto{t_n} \alpha_n$
the \emph{run assignment} $\nu(\rho)$ has domain $\bigcup_{i=0}^n V_i$ and sets $\nu(\vec V_i) := \alpha_i(\vec V)$, for all $v\,{\in}\,V$ and $0\,{\leq}\,i\,{\leq}\,n$.
\item
For an assignment $\nu$ such that $\nu \models_M H(\sigma,w)$ %$(M,\nu)$ $\TT$-satisfies $H(\sigma,w)$,
the \emph{decoded run} $\rho(M, \nu, \sigma)$ is the sequence
$\alpha_0
\goto{t_1} \alpha_1
\goto{t_2} \dots
\goto{t_n} \alpha_n$
where $\alpha_i(\vec V)\,{:=}\,\nu(\vec V_i)$, for all $v\inn V$ and $i$. %$0 \leq i \leq n$.
\end{compactitem}
\end{definition}

\begin{restatable}{lemma}{abstraction}
\label{lem:abstraction}
\begin{compactenum}[(1)]
\item
If $\BB$ admits a witness $(M,\rho)$ for $\psi$, there is a word $w$ consistent with $\rho$ and accepted by $\NFA$ such that $\nu(\rho) \models_M H(\sigma(\rho),w)$. 
\item
If $\nu \models_M H(\sigma,w)$ for some assignment $\nu$, transition sequence $\sigma$ of $\BB$, and word $w$ accepted by $\NFA$, then $\rho(M, \nu, \sigma)$ is a witness for $\psi$ of $\BB$.
\end{compactenum}
\end{restatable}
\begin{proof}[Proof (sketch)]
Both directions are by induction on $\rho$ and $\sigma$.
For the connection between a witness and $\NFA$ accepting a word consistent with $\rho$, \propref{nfa} is used.
\end{proof}

\noindent
At this point, we are ready to prove \thmref{T:witness}.
\begin{proof}[Proof (of \thmref{T:witness})]
If $(M,\rho)$ is a $\TT$-witness for $\psi$ then by \lemref{abstraction}\:(1), $(M,\nu(\rho))$ satisfies $H(\sigma(\rho),w)$ for some $w$ consistent with $\rho$ and accepted by $\NFA$.
By \remref{TT*}, there must be some $(M',\nu')$ which satisfies $H(\sigma(\rho),w)$ in $\TT^*$ as $H(\sigma(\rho),w)$ is an existential formula. By \lemref{abstraction}\:(2), $\rho(M', \nu', \sigma)$ is a $\TT^*$-witness for $\psi$.
\end{proof}

We next define our main data structure for model checking: a product automaton $\smash{\NN^\psi_\BB}$ built from the NFA $\NFA$ and the \mydds $\BB$. 
The construction uses an $\update$ function to capture how the current variable configuration, expressed as a formula $\phi$ with free variables $V$, changes by executing a transition $t$. Namely,
% 
% \begin{definition}
% \label{def:update}
% For a formula $\phi$ with free variables $V$, and a transition $t\in T$, let 
$\update(\phi, t) := \exists \vec X. \phi(\vec X) \wedge \trans{t}(\vec X, \vec V)$,
where $\vec X$ is a fresh copy of $\vec V$.
% \end{definition}
%
Below, we write $\Phi(V)$ for the set of all quantifier-free $\Sigma$-formulae over variables $V$.

\begin{definition}
\label{def:product}
Given $\BB = \dmttuple$ and the NFA $\NFA$ as above,
the \emph{product automaton} 
$\smash{\NN^\psi_\BB:=(P, R, p_0, P_F)}$ with states $P \subseteq Q \times \Phi(V)$, transition relation $R$, initial state $p_0$, and final states $P_F$ is inductively defined as follows:
\begin{compactitem}[$\bullet$]
\item $P$ contains the initial state $p_0:=(q_0, \phiinit)$; and 
\item if $(q, \varphi)\in P$ and either
\begin{compactitem}
\item[$(i)$] $(q, \varphi) = p_0$, $q_0 \goto{\varsigma} q'$ and $t=\top$, or
\item[$(ii)$]
$(q, \varphi) \neq p_0$, $q \goto{\varsigma} q'$ in $\NFA$, and  $t\in T$
\end{compactitem}
 such that the formula $\xi := \update(\varphi, t) \wedge \varsigma$ 
is $\TT$-satisfiable, there is some $(q', \varphi') \in P$ s.t.
% $\phi'$ is quantifier free, 
$\varphi' {\equiv_{\TT^*}} \xi$, and
$(q, \varphi) \goto{t,\varsigma} (q', \varphi')$ is in $R$,
\item $P_F$ consists of all $(q, \varphi)\in P$ such that $q$ is final in $\NFA$.
\end{compactitem}
\end{definition}
% 
% \begin{definition}
% \label{def:product}
% For  $\BB' = \langle B', b_0', \AA, T', F, V, \alpha_0, \guard\rangle$ and $\NFA$ as above,
% the \emph{product automaton} 
% $\smash{\NN^\psi_\BB=(P, \Sigma, R, p_0, P_F)}$ is as follows:
% \begin{compactitem}[$\bullet$]
% \item
% states in $P$ are triples $(b, q, \varphi)$ such that $b\,{\in}\,B'$,
% $q\,{\in}\, Q$, $\varphi\in \Phi$;
% \item
% the initial state is $p_0=(b_0', q_0, \phiinit)$; 
% \item
% the transition $(b,q, \varphi) \goto{a,\varsigma} (b',q', \varphi')$ is in $R$ iff 
% $b \goto{a} b'$ in $\BB'$, $q \goto{\varsigma} q'$ in 
% $\NFA$,
% \begin{compactitem}[$-$]
% \item 
% $\xi := \update(\varphi, a) \wedge \constr(\varsigma)$ 
% is $\TT$-satisfiable, $\phi'$ quantifier free, and $\varphi' {\equiv_{\TT^*}} \xi$,
% \item $B \setminus \{b'\}$ is disjoint from $\varsigma$, i.e.,
% $\varsigma$ is consistent with $b \goto{a} b'$, and
% \item
% $(b',q', \varphi')$ is in the set of accepting states $P_F$ iff $b'\in F$ and $q'\in Q_F$.
% \end{compactitem}
% \end{compactitem}
% \end{definition}

We add two remarks on \defref{product}:
First, the distinction of 
the two kinds of transitions $(i)$ and $(ii)$ is a technical requirement
as constraints in the property $\psi$ are evaluated in states of a run but constraints in $\BB$ are on transitions. %, so the latter need to be offset by one step.
To achieve this ``offset'', we use a dummy transition $t=\top$ from the initial state.
Second, a quantifier-free formula $\phi'$ as required by the second item of \defref{product} exists because $\TT^*$ has QE.

% While the construction may appear similar as the approach in~\cite{FMW22a}, there are important differences: we consider a general theory $\TT$ and all constructed formulas are enforced to be quantifier-free by using the quantifier elimination procedure of $\TT^*$;
% \todo{Sarah: keep this?}
% The correctness proof below will show that paths to final states 
% not only correspond to $\TT^*$-witnesses but also to witnesses in $\TT$.
% 
% 
We now consider \emph{paths} $\pi$ in $\smash{\NN_\BB^\psi}$ that start in the initial state and thus have, for some $n\geq 0$, the form
\begin{equation}
\label{eq:path}
% \pi\colon  
(q_0, \phiinit) {\goto{\top,\varsigma_0}} (q_1, \phi_1) {\goto{t_1,\varsigma_1}} \dots {\goto{t_n,\varsigma_n}} (q_{n+1},\phi_{n+1})
\end{equation}
i.e., the first edge is labeled $(\top, \varsigma_0)$ for some $\varsigma_0$.
% due to the two kinds of transitions $(i)$, $(ii)$ in 
We write $\sigma_\pi = \langle t_1, \dots, t_n\rangle$ for the transition sequence of $\BB$
and $w_\pi$ for the word $\langle\varsigma_0,\dots,\varsigma_{n}\rangle$ read from the edge labels of $\pi$.
Our  main result relates paths to final states to witnesses:

\begin{theorem}
\label{thm:model:checking}
\begin{compactenum}
\item [(1)]
% If $\smash{\NN_\BB^\psi}$ has a 
For every path $\pi$ to a final state of $\smash{\NN_\BB^\psi}$,  
$H(\sigma_\pi, w_\pi)$ is satisfiable by some $(M,\nu)$ for $M\in \TT$,
and $\rho(M, \nu,\sigma_\pi)$ is a $\TT$-witness for $\psi$.
\item[(2)] If $\BB$ admits a $\TT$-witness for $\psi$ then $\smash{\NN_\BB^\psi}$ has a final state.
% \todo{statement (2) could be more precise: every witness can be found}
\end{compactenum}
\end{theorem}

\begin{example}
\label{exa:product}
The product of $\BB$ from \exaref{simple} with the NFA $\NFA$ from \exaref{nfa} for $\psi := (x\,{\geq}\,0) \U (s\,{=}\,\m o_2 \wedge x\,{=}\,4)$ is as shown next.
% , with the initial state on the top left:
We abbreviate $\varsigma_0 = \{x\geq 0\}$ and $\varsigma_1 = \{s{=}\m o_2, x{=}4\}$.\\
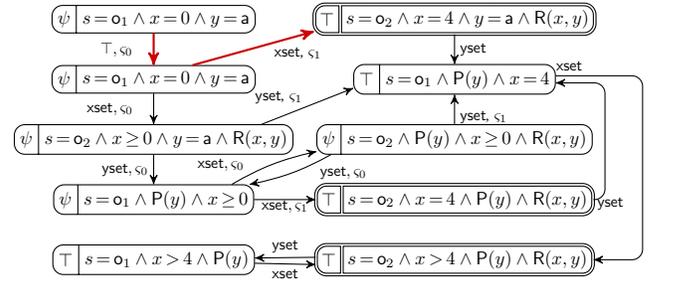
\begin{tikzpicture}[node distance=25mm,>=stealth', yscale=.8]
\tikzstyle{node} = [draw,rectangle split, rectangle split parts=2,rectangle split horizontal, rectangle split draw splits=true, inner sep=3pt, scale=.7, rounded corners]
\tikzstyle{goto} = [->]
\tikzstyle{accepting path} = [red!80!black, line width=.8pt]
\tikzstyle{accepting state} = [double]%[fill=red!80!black!15]
\tikzstyle{action}=[scale=.6, black]
\tikzstyle{constr}=[scale=.5, black]
\node[node] (0) at (0,0)  {\pcnoder{$\psi$}{$s \eqn \m o_1 \wedge x\,{=}\,0 \wedge y\,{=}\,\m a$}};
\node[node] (1) at (0,-1) {\pcnoder{$\psi$}{$s \eqn \m o_1 \wedge x\,{=}\,0 \wedge y\,{=}\,\m a$}};
\node[node] (2) at (0,-2) {\pcnoder{$\psi$}{$s \eqn \m o_2 \wedge x\,{\geq}\,0 \wedge y\,{=}\,\m a \wedge \m R(x,y)$}};
\node[node, accepting state]  (3) at (4,0) {\pcnoder{$\top$}{$s \eqn \m o_2 \wedge x\,{=}\,4 \wedge y\,{=}\,\m a \wedge \m R(x,y)$}};
\node[node] (4) at (0,-3) {\pcnoder{$\psi$}{$s \eqn \m o_1 \wedge \m P(y) \wedge x\,{\geq}\,0$}};
\node[node,] (6) at (4,-2) {\pcnoder{$\psi$}{$s \eqn \m o_2 \wedge \m P(y) \wedge x\,{\geq}\,0 \wedge \m R(x,y)$}};
% \node[node ] (7) at (8,0) {\pcnoder{$\top$}{$s \eqn 1 \wedge \m P(y) \wedge x\,{\geq}\,0 \wedge x\,{=}\,4$}};
\node[node] (8) at (4,-1) {\pcnoder{$\top$}{$s \eqn \m o_1 \wedge \m P(y) \wedge x\,{=}\,4$}};
\node[node,  accepting state] at (4,-4) (9) {\pcnoder{$\top$}{$s \eqn \m o_2 \wedge x\,{>}\,4 \wedge \m P(y)\wedge \m R(x,y)$}};
\node[node] at (0,-4) (10) {\pcnoder{$\top$}{$s \eqn \m o_1 \wedge x\,{>}\,4 \wedge \m P(y)$}};
\node[node, accepting state] (5) at (4,-3)  {\pcnoder{$\top$}{$s \eqn \m o_2 \wedge x\,{=}\,4 \wedge \m P(y)\wedge \m R(x,y)$}};
\draw[goto, accepting path] (0) to node[action, left=2mm]{$\top, \varsigma_0$} (1);
\draw[goto] (1) to node[action, left=2mm]{$\setx, \varsigma_0$} (2);
\draw[goto, accepting path] (1.20) to node[action, right=2mm, yshift=-1mm]{$\setx$, $\varsigma_1$} (3.185);
\draw[goto] (2) to node[action, left]{$\sety, \varsigma_0$} (4);
\draw[goto, bend left=10] (4.11) to node[action, left=3mm]{$\setx, \varsigma_0$} (6.185);
\draw[goto, bend left=10] (6.187) to node[action, right=3mm]{$\sety, \varsigma_0$} (4.9);
\draw[goto] (4) to node[action, below,yshift=1mm]{$\setx,\varsigma_1$} (5);
\draw[goto] (6) to node[action, right, near start]{$\sety$, $\varsigma_1$} (8);
\draw[goto,rounded corners] (5) -- node[action, right=-1mm, yshift=-1mm]{$\sety$} ($(5) + (2,0)$)   |- (8.-2);
\draw[goto] (2.8) to node[action, left=3mm, near end]{$\sety$, $\varsigma_1$} (8.185);
\draw[goto,rounded corners] (8.2) node[action, above, xshift=3mm, yshift=0mm]{$\setx$} -|  ($(9) + (2.5,0)$) -- (9);
\draw[goto] (3) to node[action, right]{$\sety$} (8);
\draw[goto] (9.179) to node[action, above]{$\sety$} (10.1);
\draw[goto] (10.-2) to node[action, below]{$\setx$} (9.182);
\end{tikzpicture}\\
% (As model checking requires to check paths to accepting states, nodes with the sink state $\bot$ of $\NFA$ are omitted.)
By \thmref{model:checking} (1), witnesses can be obtained from all paths to
accepting states. E.g., from the path $\pi$ shown in red we extract the single-step
transition sequence $\sigma = \langle \setx\rangle$ and the word
$w = \langle \{x\,{\geq}\,0\}, \{s{=}\m o_2,x{=}4\}\rangle$, with $H(\sigma,w)$ as shown below
\defref{H}; it is $\TT$-satisfied by any model $M$ s.t. $(4, \m a) \in \m R_M$ and
$\nu(s_0)\eqn\m o_1$, $\nu(s_1)\eqn\m o_2$, $\nu(x_0)\eqn 0$, $\nu(x_1) \eqn 4$, and $\nu(y_0)\eqn\nu(y_1) \eqn \m a$. The witness $\rho(M, \nu, \sigma)$ is thus
$\{s\eqn \m o_1, x\eqn 0, y\eqn\m a\} \goto{\setx} \{s\eqn \m o_2, x\eqn 4, y\eqn\m a\}$.
% $(\m 1, \domino{x}{0}{y}{\m a}) \goto{\m{a}_1}(\m 2, \domino{x}{4}{y}{\m a})$.
\end{example}

In general, the product construction $\smash{\NN_\BB^\psi}$ from \defref{product} can be infinite,
but 
\thmref{model:checking} gives rise to a semi-decision procedure:
$\smash{\NN_\BB^\psi}$ can be approximated by expanding \defref{product} in a fair way. Then an accepting path $\pi$ will be detected if it exists, so that a witness can be constructed.

\paragraph{Correctness and termination.}
To prove~\thmref{model:checking}, and show that the semi-decision procedure actually serves as a decision procedure in notable cases, we employ some additional notions. We do so using ideas from \cite{FMW22a}, with the key difference that reasoning is delegated to $\TT^*$, and moreover
we have no control states.%\todo{Check se abbastanza distinto da previous work} 

Let $\sigma = \langle t_1, \dots, t_n\rangle$ be a transition sequence and $w\in \Theta^*$ have length $n{+}1$. We 
define the \emph{history constraint} %
% \footnote{$\hist(\sigma,w)$ corresponds to the history constraint $h(\sigma,\vec c)$ in \cite{FMW22a}.} 
$\hist(\sigma,w) := (\exists \vec V_0 \dots \vec V_{n-1}.\: H(\sigma,w))(\vec V)$, which is a formula with free variables $V$.
% 
% \begin{definition}\label{def:history constraint}
% Let $\sigma\colon b_0 \goto{a_1} b_1 \goto{a_2} \dots \goto{a_n} b_n$ be a transition sequence of $\BB$, $\sigma'$ its $n{-}1$-step prefix
% and $\vec c = \langle c_0, c_1, \dots, c_n\rangle$ a sequence of $\Sigma$-constraints.
% The \emph{history constraint} $h(\sigma, \vec c)$ 
% is inductively defined as
% $h(\sigma, \vec c) = \phiinit \wedge c_0$ if $n\,{=}\,0$, and 
% $h(\sigma, \vec c) = \update(h(\sigma', \langle c_0, \dots,c_{n-1}\rangle), a_{n}) \wedge c_n$ if $n>0$.
% \end{definition}
% 
% \noindent
% Satisfying assignments of history constraints are closely related to runs:
% 
% \begin{lemma}
% \label{lem:abstraction}
% Let $\sigma$ be a transition sequence of $\BB$, %of length $n$, 
% $\vec c = \langle c_0, \dots, c_n\rangle$, and $\alpha$ a state variable assignment.
% For any $\TT^*$-model $M$,
% $\alpha \models_M \smash{h(\sigma,\vec c)}$ iff there is a
% $\TT^*$-run $(M,\rho)$ such that 
% $\sigma$ abstracts $\rho$, and if
% $\rho$ is of the form
% $\smash{(b_0, \alpha_0) \goto{a_1} \dots \goto{a_n} (b_n, \alpha_n)}$
% then $\alpha=\alpha_n$, and
% $\alpha_i \models_M c_i$ for all $0\leq i \leq n$.
% \end{lemma}
% 
% The straightforward induction proof can be found in the appendix.
%
% For $\Theta$ the alphabet of the NFA and a word $w = \langle\varsigma_0,\dots,\varsigma_n\rangle\in \Theta^*$, let
% $\vec c(w) = \langle\constr(\varsigma_0), \dots, \constr(\varsigma_n)\rangle$
% be the respective constraint sequence.
The next result formalizes in which sense
$\smash{\NN_\BB^\psi}$ is a
product construction: a path $\pi$ of the form \eqref{eq:path} combines $\sigma_\pi$ with $w_\pi$
and the last formula in $\pi$ is equivalent to the history constraint $\hist(\sigma_\pi,w_\pi)$.

\begin{restatable}{lemma}{productlemma}
\label{lem:product}
$\smash{\NN_\BB^\psi}$ has a non-empty path $\pi$ to a node $(q,\phi)$ iff
$\NN_\psi$ has a transition sequence ending in $q$ labeled $w_\pi$ and
$\BB$ a transition sequence
$\sigma_\pi$ such that
$\phi \equiv_{\TT^*} \hist(\sigma_\pi,w_\pi)$ and $\phi$ is $\TT^*$-satisfiable.
% \begin{compactenum}
% \item[(1)]
% If $\smash{\NN_\BB^\psi}$ has a path $\pi$ of length $n\geq 1$ to a node $(q,\phi)$, then
% $\NN_\psi$ has a transition sequence ending in $q$ labeled $w_\pi$,
% $\BB$ a transition sequence
% $\sigma_\pi$ consistent with $w_\pi$, and
% $\phi \equiv_{\TT^*} \hist(\sigma_\pi,w_\pi)$ is $\TT^*$-satisfiable.
% \item[(2)]
% If $\NN_\psi$ has a transition sequence ending in $q$ labeled $w$, $\BB$ a
% transition sequence $\sigma$ that is consistent with $w$, and
% $\hist(\sigma,w)$ is $\TT^*$-satisfiable, then $\smash{\NN_\BB^\psi}$ has a path $\pi$ to $(q,\phi)$ s.t. $w=w_\pi$, $\sigma=\sigma_\pi$, and
% $\phi \equiv_{\TT^*} \hist(\sigma,w)$.
% \end{compactenum}
\end{restatable}
\noindent
Both directions are by plain induction.
Now we can show: %prove \thmref{model:checking}:
\begin{proof}[Proof (of \thmref{model:checking})]
(1)
Let $\pi$ be a path to a final state $(q_n,\phi)$ in $\smash{\NN_\BB^\psi}$.
So $q_n$ is final in $\NFA$, and we cannot have  $q_n=q_0$  
as $\psi \neq \top$.
By \lemref{product} ($\Longrightarrow$), there is a transition sequence
in $\NN_\psi$ labeled $w_\pi=\langle\varsigma_0 \dots\varsigma_n\rangle$ and $\hist(\sigma_\pi,w_\pi)$ is $\TT^*$-satisfiable.
So also $H(\sigma_\pi,w_\pi)$ must be $\TT^*$-satisfiable,
and by \remref{TT*} also $\TT$-satisfied by some $M\in \TT$ and $\nu$.
By \lemref{abstraction} (2), $\rho(M, \nu, \sigma)$ is a witness for $\psi$. 
% of $\BB$.
% 
% \noindent
(2) 
If there is a $\TT$-witness for $\psi$, by \thmref{T:witness}, there is a $\TT^*$-witness $(M,\rho)$.
By \lemref{abstraction} (1), there is a word $w$ consistent with $\rho$ and accepted by $\NFA$ such that $H(\sigma(\rho),w)$ is $\TT$-satisfiable. 
By \lemref{product} ($\Longleftarrow$), the accepting transition sequence of $\NN_\psi$ labeled $w$ and $\sigma(\rho)$ give rise to a path $\pi$ in $\smash{\NN_\BB^\psi}$
such that $w\eqn w_\pi$ and $\sigma(\rho)\eqn \sigma_\pi$.
As $\NN_\psi$ accepts $w$, $\pi$ leads to a final state.
\end{proof}

\noindent
We next state an abstract decidability criterion.
% Let $\hist(\sigma, w)$ be a history constraint \emph{of} $\BB$ and $\psi$ if $\sigma$ is a transition sequence of $\BB$ and $w\in \Theta^*$.
% , for $\Theta$ the alphabet of $\NFA$.

\begin{definition}
\label{def:history:set}
A \emph{data history} for $(\BB, \psi)$ is a minimal
set of quantifier-free formulae $\Phi$ such that
for every 
% history constraint $\hist(\sigma,w)$ of $\BB$ and $\psi$ 
transition sequence $\sigma$ of $\BB$ and $w\in \Theta^*$,
there is some $\phi \inn \Phi$ such that $\hist(\sigma,w)\,{\equiv_{\TT^*}}\, \phi$.
\end{definition}

Intuitively, a \historyset has a formula representation of every pair of a transition sequence in $\BB$, and a sequence of constraints needed to verify $\psi$.
As formulas in $\smash{\NN_\BB^\psi}$ are history constraints (\lemref{product}), 
if $(\BB, \psi)$ has \theproperty then a finite $\smash{\NN_\BB^\psi}$ exists, so
by \thmref{model:checking} we have decidability:

\begin{corollary}
\label{cor:main}
The verification task is decidable
if $(\BB, \psi)$ has a \theproperty.
\end{corollary}

\noindent
%
% Note that since linear arithmetic has quantifier elimination, \corref{main} generalizes the results from~\cite{FMW22a} about arithmetic DDSs with monotonicity and integer periodicity constraints, for which the existence of a finite history set was shown.
% \todo{S: how about finite summary version?}
%
In the next section, we use \corref{main} to identify concrete decidable classes of DMTs that involve databases and arithmetic, but
our approach is not limited to this domain. This is illustrated by \exaref{lists} in the appendix, it uses the theory of lists.

\section{Decidability Criteria}
\label{sec:decidable}

We now use \corref{main} to identify concrete, checkable classes of {\plainmydds}s where our verification task is decidable.
We focus on systems that query a read-only database, which are
crucial to BPM and database processes \cite{BojanczykST13,CGM13,DHLV18};
and extend this setting with arithmetic,
a key combination in practice.
Note that we address a parameterized verification problem, checking whether a witness exists for some instantiation of the initial DB.

% Sarah added
Let a \emph{\DBDDS} be a \mydds where $\Sigma$ consists of arbitrary relation, and unary function symbols, and let its theory $\TT_{db}$ be EUF possibly extended with a set of ground facts to model an initial DB. This theory %is denoted EUF$^+$ below; it
is an expressive model for databases,
where unary function symbols can capture primary keys; and it
enjoys model completion~\cite{CGGMR20}. 
% Notably, \DBDDS{s} can model \emph{simple artifact systems} (SASs)~\cite{CGGMR20},
% where unary function symbols conceptually capture primary keys.
%
In a \emph{\DBDDS with arithmetic}, $\Sigma$ supports in addition operations from an arithmetic theory $\TT_{ar}$ such as LIRA,
% We assume that such a system is evaluated over a %universal 
% $\Sigma$-theory $\TT:=\TT_{db}\cup \TT_{ar}$ (recall that $\TT_{db}$ can be EUF$^+$).
and its theory is $\TT:=\TT_{db}\cup \TT_{ar}$.
\exasref{intro}{simple} are examples for {\DBDDS}s with arithmetic.
To ensure that the combined theory has model completion, we assume that the signature is \emph{tame},
which intuitively means that no uninterpreted function symbol has an arithmetic domain (see below).

% Sarah added
In general, both the DB and the arithmetic perspective can give rise to infinitely many terms, so in classes of \mydds{s} where the verification task is decidable, both must be suitably restricted. Below, we explore different ways to do so.
% \todo{Si capisce l'aspetto parameterized? HO detto una frasetta in introduzione}
% \todo{S: added}

\paragraph{I. Acyclic signature.}
First, we consider {\DBDDS}s without arithmetic.
% over an \emph{acyclic} (also called \emph{stratified}) 
% signature $\Sigma$~\cite{ARS2010}
Acyclicity of the signature $\Sigma$ is defined via the \emph{sort graph} of $\Sigma$, which
% The sort graph 
has as node set the sorts $\SS$, and an edge from $s$ to $s'$ iff there is a function symbol $f\colon s \to s'$ in $\Sigma$ (as we have only unary functions); 
we say that $\Sigma$ is \emph{acyclic} if the sort graph is so. In this case, only finitely many $\Sigma$-terms exist~\cite{ARS2010}. Thus, there are also only finitely many non-equivalent quantifier-free formulas over the finite set of variables $V$,
so finitely many history constraints. Hence if $\Sigma$ is acyclic, every \DBDDS has a finite \historyset, so by \corref{main}:

\begin{theorem}
\label{thm:acyclic:SAS}
The verification task for a \DBDDS is decidable if $\Sigma$ is acyclic.
\end{theorem}

This generalizes the result about decidability of reachability in SASs over an acyclic $\Sigma$ \cite[Thm 4.2]{CGGMR20}, as we consider here full LTL$_f$ model checking.
An example for a respective process is shown in \exaref{acyclic}.

\paragraph{II. Acyclic signature and monotonicity constraints.}
Next, we consider {\DBDDS}s with arithmetic $\BB$.
As mentioned above, in the sequel we assume  that signature $\Sigma$ is \emph{tame}, 
i.e., in the sort graph of $\Sigma$, sorts $\mathit{rat}$ and $\mathit{int}$ are leaves.
In this case, the combined theory $\TT$ is known to enjoy
model completion ~\cite[Thm. 7]{CalvaneseGGMR22}.
E.g., the signatures in \exasref{intro}{simple} are tame.
%
% To obtain decidability, 
Moreover, towards decidability
we restrict arithmetic constraints to
\emph{monotonicity constraints} (MCs) - constraints of the form $t \odot t'$ where $t,t'$ are $\Sigma$-terms of sort $rat$ over free variables $V$ 
and $\odot$ is one of $=$, $\neq$, $\leq$, or $<$.
MCs have been often considered in the literature~\cite{DD07,FLM19} and are important in BPM as guards of this shape can be learned from data~\cite{LeoniA13}.
% Second, \emph{integer periodicity constraints} (IPCs)~\cite{Demri06} have the form $t = t'$, $t \odot d$ for $\odot \in \{=,\neq, <, >\}$, $t \equiv_k t + d$, or $t \equiv_k d$,
% where $t,t'$ are $\Sigma$-terms of sort $int$ over free variables $V$,
% $d\in \mathbb N$, and $\equiv_k$ denotes equality modulo $k\in \mathbb N$. 
% 
\begin{restatable}{theorem}{theoremacyclicMC}
\label{thm:MC:acyclic}
Let $\Sigma$ be acyclic,
$\BB$ a \DBDDS with arithmetic %having  $\TT_{db}:=EUF$  
and $\psi \in \LL_\Sigma$.
If all arithmetic constraints in $\BB$ and $\psi$ are MCs, %or all IPCs, 
the verification task is decidable.
\end{restatable}
\begin{proof}[Proof (idea)]
We inspect procedure $\m{TameCombCover}$ \cite[Sec. 8]{CalvaneseGGMR22} showing that it has finitely many possible results of the form $\phi_1(\vec X) \wedge \phi_2(\vec Y, \vec t(\vec X))$, up to equivalence. Here $\phi_1$ is an EUF formula, $\vec t$ is a list of EUF terms, and $\phi_2$ an arithmetic formula. Acyclicity bounds the possibilities for $\phi_1$ and $\vec t$, and by results about QE of MCs, there are only finitely many choices for $\phi_2(\vec Y, \vec t(\vec X))$\cite[Thm. 5.2]{FMW22a}.
\end{proof}

\noindent
E.g., \thmref{MC:acyclic} applies to \exaref{simple}, but not to \exaref{intro} (not all constraints are MCs).
\thmref{MC:acyclic} generalizes \cite[Thm. 5.2]{FMW22a} as it supports arithmetic \emph{and} EUF.

\paragraph{III. Local Finiteness.}
Next, we consider decidability beyond acyclic $\Sigma$ for \DBDDS without arithmetic, using the concept of local finiteness~\cite{Lipparini82}.
A $\Sigma$-theory $\TT$ is called \emph{$k$-locally finite} if  for any finite set of variables $X$, the number of $\Sigma$-terms with free variables $X$ is upper-bounded by some $k$, up to $\TT$-equivalence.
As  local finiteness directly ensures that there are only finitely many non-equivalent quantifier-free formulas in $\TT^*$, we readily get:

\begin{theorem}
\label{thm:locally:finite}
If $\TT$ is a $k$-locally finite $\Sigma$-theory that admits model completion, the $\TT$-verification task is decidable for any \DBDDS $\BB$ and
 $\Sigma$-property $\psi$.
\end{theorem}

This generalizes \cite[Thm. 5]{BojanczykST13} to LTL$_f$, since the Fra\"{\i}ss\'{e} classes from  \cite{BojanczykST13}, when FO definable, coincide with universal, locally finite theories admitting a model completion \cite{CGGMR20}.
One way to ensure that a model completion exists in \thmref{locally:finite} is by restricting axioms in $\TT$ to single-variable 
sentences: %\cite{CGGMR20}:

\begin{corollary}
\label{cor:locally:finite}
The $\TT$-verification task is decidable for a \DBDDS $\BB$ and
 $\Sigma$-property if $\TT$ is a $k$-locally finite $\Sigma$-theory axiomatized by single-variable axioms.
\end{corollary}

 \thmref{locally:finite} generalizes \thmref{acyclic:SAS} as acyclicity implies local finiteness but not vice versa: consider a signature 
where $f$ is the only non-constant function symbol and
$\TT := \{\forall x.f^2(x)=x\}$, which is locally finite but not 
% \todo{S: example in comments}
acyclic.
One can also obtain decidability by replacing acyclicity in the combination result of \thmref{MC:acyclic} by local finiteness, using in 
the proof local finiteness to ensure that up to equivalence there are only finitely many formulas without arithmetic.

\paragraph{IV. Bounded lookback.}
The above criteria achieve decidability mostly by syntactic restrictions. 
However, \corref{main} can also be applied to arbitrary constraints, by controlling how transitions interact with each other.
This is the case for systems with \emph{bounded lookback}, a property
intuitively expressing that the behaviour of a \plainmydds depends only on a bounded amount of information from the past~\cite{FMW22a}; it generalizes the feedback freedom property~\cite{DDV12}.
%
%
% Indeed, we obtain decidability also for \DBDDS with arithmetic that enjoy bounded lookback, which applies e.g. to \exaref{intro}; however, for reasons of space this result had to be moved to the appendix (cf. \secref{bounded:lookback}).
%
Bounded lookback is defined via \emph{computation graphs}:
Let $\sigma = \langle t_1, \dots, t_n\rangle$ be a transition sequence of a
\DBDDS with arithmetic $\BB$, and $w= \langle \varsigma_0, \dots, \varsigma_n\rangle\in 2^C$.
The \emph{computation graph} for $\sigma$ and $w$ is an undirected graph
$G_{\sigma,w}$ with nodes
$\mc V := \{v_i \mid v\in V\text{, }0\,{\leq}\,i\,{\leq}\,n\}$
and an edge $(x_i, y_j)$ whenever $x_i$ and $y_j$ are in the transitive closure of variable pairs that occur in a common literal of 
$H(\sigma,w)$, for all $i,j \leq n$ (see the appendix for examples and details).
For $k\geq 0$, $(\BB, \psi)$ has \emph{$k$-bounded lookback} if for all $\sigma$ and all $w\in 2^C$ s.t. $H(\sigma,w)$ is $\TT$-satisfiable, all acyclic paths in $[G_{\sigma,w}]$ have length at most $k$.
E.g., one can check that \exaref{intro} has 5-bounded lookback, as variables get reset when the loop repeats,
so the dynamics depends only on the last five steps.
Also DMTs without loops have bounded lookback, since the length of computation graphs is bounded.
We show:

\begin{restatable}{theorem}{theoremboundedlookback}
\label{thm:bounded:lookback}
If $\BB$ is a \DBDDS with arithmetic such that
 $(\BB, \psi)$ has $k$-bounded lookback for some $k\geq 0$, the verification task is decidable.
\end{restatable}
% \todo{Marco: ho rimosso il paragrafo che parla della prova}
%The proof uses the fact that history constraints correspond to formulas with a bounded quantifier depth over a finite vocabulary, of which there are only finitely many.

\section{Implementation}
\label{sec:implementation}

% \paragraph{Tool.}
We implemented our %model checking
 procedure for verifying  {\DBDDS}s with arithmetic in a tool, called %called 
 \textsc{LinDMT}.
Given a property and a \DBDDS as input,
the tool 
% checks membership in one of the decidable classes identified by \thmsref{MC:acyclic}{bounded:lookback},
visualizes the NFA and product construction,
and computes a witness as prescribed by \thmref{model:checking}, if it exists.
The tool interfaces %Z3~\cite{Z3} and 
CVC5~\cite{DRKBT14} for SMT checks, and QE in LIRA.
On top of that, we implemented the cover computation for EUF with unary functions and arbitrary relations from~\cite[p. 961]{CalvaneseGGMR21}, 
%to be able to handle {\DBDDS}s with arithmetic,
and the $\m{TameCombCover}$ routine \cite[Sec. 8]{CalvaneseGGMR22} to compute covers in tame combinations of EUF and LIRA. 
The tool is written in Python, and accessible via a web interface
%(\url{https://tinyurl.com/mr2nuf62}),
(\url{http://18.117.176.211/lindmt}),
where also sources and benchmarks are available.
In input, control states (as $\m{o}_1$, $\m{o}_2$ in \exaref{simple}) are supported for efficiency.

For a preliminary evaluation, we created a set of benchmarks using data-aware processes from the literature, including problems adapted from~\cite{verifas}, excluding updates of DB relations. % which are not supported by DMTs. 
For each DMT, we checked five different properties.
The results are summarized in the table below,
listing whether the problem falls into a decidable class I--IV,
the number of transitions (\textbf{T}), 
total size of transition formulae (\textbf{D}), 
number of relations/functions/constants (\textbf{R}/\textbf{F}/\textbf{C}),
the total/average/maximal time in seconds for the five properties, 
and the total/average/maximal number of SMT checks.
% For experiments, we represented a set of data-aware processes from the literature as DMTs,
% including the VERIFAS~\cite{verifas} problems (which were adapted from BPMN examples~\cite{bpmn}). Some of them had to be adapted to exclude updates of
% DB relations, as such updates are not supported by DMTs, but we can verify \LTLf properties of these problems that were beyond reach of VERIFAS.

\noindent
\begin{footnotesize}
\begin{tabular}{@{}l@{\ \ }c@{\ \ }c@{\ \ }c@{\ \ }c@{/}c@{/}c@{\:}c@{/}c@{/}c@{\ \ }c@{/}c@{/}c@{}}
 & &\textbf{T} &\textbf{D} & \textbf{R}&\textbf{F}&\textbf{C} & \multicolumn{3}{c}{\textbf{time}} & \multicolumn{3}{c}{\textbf{SMT checks}} \\ \hline
  (P1) & % bpm_credit_approval.json
    IV & 
    10 & 60 & 0 & 0 & 9 &
    3.47 & 0.69 & 1.23 &
    2702 & 540 & 880
 \\ \hline
  (P2) & % bpm_hospital_billing.json
    IV & 
    44 & 124 & 0 & 0 & 7 &
    14.23 & 2.85 & 3.66 &
    20470 & 4094 & 4285
 \\ \hline
  (P3) & % bpm_order_process.json
    IV & 
    6 & 30 & 0 & 0 & 7 &
    0.95 & 0.19 & 0.40 &
    186 & 37 & 57
 \\ \hline
  (P4) & % bpm_package_handling.json
    II & 
    39 & 257 & 0 & 0 & 10 &
    8.14 & 1.63 & 2.32 &
    3811 & 762 & 1342
 \\ \hline
  (P5) & % bpm_road_fines_mined.json
    IV & 
    19 & 37 & 0 & 0 & 3 &
    1.71 & 0.34 & 0.71 &
    1130 & 226 & 433
 \\ \hline
  (P6) & % bpm_road_fines_normative.json
    IV & 
    20 & 88 & 0 & 0 & 3 &
    5.43 & 1.09 & 2.40 &
    2816 & 563 & 1109
 \\ \hline
  (P7) & % paper_simple.json
    II & 
    2 & 9 & 2 & 0 & 1 &
    0.31 & 0.06 & 0.11 &
    127 & 25 & 33
 \\ \hline
  (P8) & % paper_webshop.json
    IV & 
    8 & 52 & 2 & 1 & 2 &
    1.36 & 0.27 & 0.60 &
    506 & 101 & 239
 \\ \hline
  (P9) & % verifas_acquisition_rfc.json
    IV & 
    20 & 194 & 3 & 0 & 5 &
    8.18 & 1.64 & 2.91 &
    11363 & 2272 & 4067
 \\ \hline
  (P10) & % verifas_airline_checkin.json
    IV & 
    20 & 159 & 2 & 1 & 3 &
    54.23 & 10.85 & 17.38 &
    17648 & 3529 & 4265
 \\ \hline
  (P11) & % verifas_amazon.json
    IV & 
    19 & 196 & 5 & 0 & 3 &
    3.30 & 0.66 & 0.93 &
    568 & 113 & 288
 \\ \hline
  (P12) & % verifas_book_writing.json
    IV & 
    8 & 106 & 2 & 0 & 4 &
    3.03 & 0.61 & 0.79 &
    2319 & 463 & 663
 \\ \hline
  (P13) & % verifas_hardware.json
    IV & 
    21 & 291 & 2 & 0 & 7 &
    13.31 & 2.66 & 2.98 &
    8144 & 1628 & 1829
 \\ \hline
  (P14) & % verifas_lasertec.json
    IV & 
    12 & 264 & 1 & 0 & 4 &
    17.43 & 3.49 & 6.26 &
    2858 & 571 & 804
 \\ \hline
  (P15) & % verifas_mortgage.json
    IV & 
    16 & 107 & 2 & 6 & 6 &
    2.70 & 0.54 & 0.98 &
    375 & 75 & 171
 \\ \hline
  (P16) & % verifas_newcar.json
    IV & 
    16 & 213 & 3 & 0 & 4 &
    3.11 & 0.62 & 0.75 &
    488 & 97 & 138
 \\ \hline
  (P17) & % verifas_property_casuality.json
    II & 
    10 & 22 & 2 & 4 & 5 &
    0.38 & 0.08 & 0.13 &
    134 & 26 & 46
 \\ \hline
  (P18) & % webservice_loan.json
    IV & 
    8 & 35 & 1 & 1 & 7 &
    6.29 & 1.26 & 5.69 &
    560 & 112 & 195
 \\ \hline
\end{tabular}
\end{footnotesize}
All problems fall in one of the decidable classes.
Overall, about 36\% of the computation time is spent on SMT checks, and 29\% on quantifier elimination/cover computation.
% However, the tool is intended as a prototype, and extensive experiments are left for future work.

% \section{Conclusion}
% \label{sec:conclusion}

\section{Conclusion}
We have tackled linear-time verification for {\plainmydds}s, a very general model of data-aware processes.
% where datatypes 
%over a decidable theory with model completion.
% We considered {\plainmydds}s, a very general model of transition systems operating over a decidable theory with model completion, and presented a linear-time verification procedure for such systems.
Our verification procedure is semi-decision in general, but becomes a decision procedure for several relevant, checkable settings. We have shown novel decidability criteria
especially for processes operating over read-only databases with arithmetic,
generalizing several earlier results.
% generalizing several earlier works ~\cite{DDV12,FMW22a,CGGMR20,BojanczykST13}.
% Notably, it has practically relevant implications for the verification of processes operating over databases with arithmetic.
%
In the future, we plan to study other %important 
verification tasks for {\plainmydds}s, such as branching-time model checking and monitoring.
Second, we want to investigate whether our techniques can be applied to evolving, read-write databases. Third, we want to refine our results by considering description logics that admit a model completion ~\cite{Koopmann20,BaaderGL12}.
% Second, we want to investigate whether our decidability criterion can be relaxed to a \emph{finite summary}~\cite{FMW22a}.
% Third, our results also cover background theories that admit a model completion, e.g. certain description logics~\cite{Koopmann20,BaaderGL12}; we want to further study these settings. 
%ALC~\cite{DL-CGMM21,CateCMV06}.
%Third, our results also cover e.g. certain description logics that admit a model completion, we want to further study these settings.

%\clearpage
\bibliography{references}

\appendix
\clearpage
\section{Automata Construction}

For a
\mydds $\BB = \dmttuple$ and verification property $\psi\in \LL_\Sigma$, let $C$ be the set of all constraints in $\psi$.
The construction of $\NFA$ follows earlier work~\cite{dGV13,GMM14,FMW22a}, with the difference that we allow general constraints in the LTL$_f$ property, including existential quantification, and since $\TT$ need not have quantifier elimination, some care must be taken.

The alphabet of $\NFA$ is $\smash{\Theta=2^{C}}$, i.e., a symbol is a subset of constraints in $C$. The construction
also employs an extended alphabet $\smash{\Theta_\last=2^{C \cup \{\last,\neg\last\}}}$, where $\last$ is an auxiliary proposition to mark the last element of a trace.
We call $\varsigma\inn \Theta_\last$ \emph{satisfiable} if 
$\bigwedge\varsigma$ is $\TT$-satisfiable, and
$\{\last, \neg \last \}\not \subseteq \varsigma$.

The NFA construction follows~\cite{dGV13} in that it is defined via a function $\delta$.
The input of $\delta$ is a (quoted) property $\inquotes{\phi}$ for $\phi \in \LL_\Sigma\cup \{\top,\bot\}$ with constraints in $C$; 
the output is a set of pairs
$(\inquotes{\phi'},\varsigma)$ where $\phi' \in \LL_\Sigma \cup \{\top,\bot\}$ is a property with constraints in $C$, and $\varsigma \in \Theta_\last$.

We consider the conjunction and disjunction of  two sets of such tuples $R_1$ and $R_2$, defined as follows:
\begin{align*}
R_1 \ovee R_2 = &\{(\inquotes{\psi_1{\vee}\psi_2}, \varsigma_1 \cup \varsigma_2) \mid (\inquotes{\psi_1}, \varsigma_1) \inn R_1\text{, }\\
&\quad(\inquotes{\psi_2}, \varsigma_2) \inn R_2\text{, and }\varsigma_1\cup \varsigma_2\text{ is satisfiable}\} \\
R_1 \owedge R_2 = &\{(\inquotes{\psi_1{\wedge}\psi_2}, \varsigma_1 \cup \varsigma_2) \mid (\inquotes{\psi_1}, \varsigma_1) \inn R_1\text{, }\\
&\quad(\inquotes{\psi_2}, \varsigma_2) \inn R_2\text{, and }\varsigma_1\cup \varsigma_2\text{ is satisfiable}\}
\end{align*}
where $\psi_1 \vee \psi_2$ and $\psi_1 \wedge \psi_2$ are simplified if possible.
Then the function $\delta$ is defined as follows:\\

\noindent
\begin{tabular}{@{~}l@{~}l}
$\delta(\inquotes{\top})$ &= 
  $\{(\inquotes{\top},\emptyset)\}
  \text{ and }\delta(\inquotes{\bot}) = \{(\inquotes{\bot},\emptyset)\}$\\
$\delta(\inquotes{c})$ &= 
  $\{(\inquotes{\top},\{c\}),(\inquotes{\bot},\emptyset)\} 
  \text{ if $c \in C$}$\\
$\delta(\inquotes{\psi_1 \vee \psi_2})$ &= 
  $\delta(\inquotes{\psi_1}) \ovee \delta(\inquotes{\psi_2})$ \\
$\delta(\inquotes{\psi_1 \wedge \psi_2})$ &= 
  $\delta(\inquotes{\psi_1}) \owedge \delta(\inquotes{\psi_2})$ \\
$\delta(\inquotes{\X \psi})$ &= 
  $\{(\inquotes{\psi},\{\neg \last\}), (\inquotes{\bot}, \{\last\})\}$ \\
$\delta(\inquotes{\G \psi})$ &=
  $\delta(\inquotes{\psi}) \owedge (\delta(\inquotes{\X\G \psi}) \ovee \delta_\lambda)$ \\
$\delta(\inquotes{\psi_1 \U \psi_2})$ &=
  $\delta(\inquotes{\psi_2}) \ovee (\delta(\inquotes{\psi_1})
  \owedge \delta(\inquotes{\X(\psi_1 \U \psi_2)}))$
\end{tabular}\\

\noindent
where $\delta_\lambda=\{(\inquotes{\top},\{ \lambda\}),(\inquotes{\bot},\{\neg\lambda\})\}$.
While the symbol $\last$ is needed for the construction, we can omit it from the NFA,
as done in~\cite{GMM14}, 
and define $\NFA$ as follows:
\begin{definition}
\label{def:NFA}
Given a formula $\psi \inn \LL_\Sigma$, let the NFA 
$\NFA\,{=}\,(Q, \Theta, \varrho, q_0, \{q_f, q_e\})$
be given by the initial state
$q_0\,{=}\,\inquotes{\psi}$, final state
$q_f\,{=}\,\inquotes{\top}$ and an additional final state $q_e$,
and where
$Q$ and $\varrho$ are the smallest sets such that $q_0, q_f, q_e \in Q$ and whenever 
$q\in Q\setminus\{q_e\}$ and $(q', \varsigma)\in \delta(q)$
% such that $\{\last,\neg \last\} \not\subseteq \varsigma$ % automatic 
then $q'\in Q$,
\begin{compactenum}[(i)]
\item if $\last \not\in \varsigma$ then
$(q, \varsigma \setminus\{ \neg \last\}, q') \in \varrho$, and
\item
if $\last \in \varsigma$ and $q' = \inquotes{\top}$ then
$(q, \varsigma \setminus\{\last, \neg \last\}, q_e) \in \varrho$.
\end{compactenum}
\end{definition}

To state correctness, let a word $\langle\varsigma_0,\dots,\varsigma_n\rangle\in \Theta^*$ be \emph{consistent} with a run $(M,\rho)$ for
$\smash{\rho\colon \alpha_0 
\goto{t_1} \alpha_1
\goto{t_2} \dots \goto{t_n} \alpha_n}$ if
$\alpha_i \models_M \bigwedge \varsigma_i$
for all $0\leq i \leq n$.

The proof of the correctness statement about the NFA $\NFA$ (\propref{nfa}) requires some preliminary statements.
First, we define more precise consistency notions.
For a run $(M,\rho)$ with
\begin{equation}
\label{eq:run}
\rho\colon\alpha_0
\goto{t_1} \alpha_1
\goto{t_2} \dots \goto{t_n} \alpha_n
\end{equation}
let $\varsigma \in \Theta$ be \emph{consistent with step} $i$ of $(M,\rho)$
if %$\varsigma_i$ is disjoint from $B\setminus\{b_i\}$ and
$\alpha_i \models_M \bigwedge \varsigma_i$.
Thus, a word $\varsigma_0 \varsigma_1 \cdots \varsigma_n\in \Theta^*$ is consistent
with $(M,\rho)$ if $\varsigma_i$ is consistent with step $i$ of $(M,\rho)$ for all $i$, $0\,{\leq}\,i\,{\leq}\,n$.
Moreover, let $\varsigma \in \Theta_\last$ be \emph{$\lambda$-consistent} with step $i$ of $(M,\rho)$
if $\varsigma$ is consistent with step $i$ of $\rho$, 
if $i<n$ then $\last \not\in \varsigma$, and 
if $i=n$ then $\neg \last \not\in \varsigma$.

As a first auxiliary result, we note that $\delta$ is total in the sense that
its result set contains an element that is $\lambda$-consistent with any pair of an assignment $\alpha$ and
a run $\sigma$:

\begin{lemma}
\label{lem:delta total}
Let $\psi \in \LL_\Sigma \cup \{\top,\bot\}$ be a property,
$(M,\rho)$ a run such that $\rho$ has length $n$, and $i$ such that $0\leq i \leq n$.
Then there is some $(\inquotes{\psi'}, \varsigma) \in \delta(\inquotes{\psi})$ 
such that 
$\varsigma$ is $\lambda$-consistent with step $i$ of $\rho$.
\end{lemma}
\begin{proof}
Let $\rho$ be of the form \eqref{eq:run}.
By structural induction on $\psi$.
If $\psi$ is $\top$ or $\bot$ then we must have $\varsigma=\emptyset$, which is trivially $\lambda$-consistent with any step.
% If $\psi$ is a state $b\in B$, then we can take $(\inquotes{\top}, \{b\}) \in \delta(\inquotes{b})$ if $b_i = b$, and $(\inquotes{\bot}, \emptyset)$ otherwise.
If $\psi$ is a constraint $c$, then we can take $(\inquotes{\top}, \{c\}) \in \delta(\inquotes{c})$ if $\alpha_i \models_M c$, and $(\inquotes{\bot}, \emptyset)$ otherwise.
If $\psi$ is of the form $\X \psi'$ then $(\inquotes{\psi'}, \{\neg \lambda\}) \in \delta(\inquotes{\psi})$ is $\lambda$-consistent if $i<n$, and
$(\inquotes{\bot}, \{\lambda\})$ otherwise.
% For the set of tuples $\{(\inquotes{\top},\{ \lambda\}),(\inquotes{\bot},\{\neg\lambda\})\}$ occurring in $\delta(\G  \dots)$, we have $\lambda$-consistency with $(\inquotes{\top},\{ \lambda\})$ if $i=n$, and with $(\inquotes{\bot},\{\neg\lambda\})$ otherwise.
In all other cases the claim follows from the definitions of $\ovee$ and $\owedge$ and the induction hypothesis.
\end{proof}

Next, one can show that $\delta$ preserves and reflects whether a run satisfies a property in $\LL_\Sigma$.
Both directions are proven by straightforward induction proofs on the property
structure, exactly like \cite[Lem. A.3]{FMW22a}.

\begin{lemma}
\label{lem:delta}
Let $\varphi \in \LL_\Sigma \cup \{\top,\bot\}$,
$(M,\rho)$ be a run with $\rho$ of the form \eqref{eq:run}, and $0\,{\leq}\,i\,{\leq}\,n$.
Then
$\rho,i \models_M \varphi$ if and only if
there is some $(\inquotes{\varphi'}, \varsigma)\in \delta(\inquotes{\varphi})$ such that\\
\noindent
\begin{tabular}{@{\ }r@{\ }p{8cm}}
$(a)$ & $\varsigma$ is $\lambda$-consistent with step $i$ of $\rho$, \\
$(b)$ &either $i<n$ and $\rho,i{+}1 \models_M \varphi'$, or $i\,{=}\,n$ and $\varphi' = \top$.
\end{tabular}
\end{lemma}

Let a word $\varsigma_0 \varsigma_1 \cdots \varsigma_{n}\in \Theta_\last^*$ be \emph{well-formed}
if $\last \not\in \varsigma_i$ for all $0\leq i < n$, and $\neg \last \not\in \varsigma_n$.
Then \lemref{delta} can be used to show that a sequence of $\delta$ applications corresponds to a word $w$ that leads to the accepting state $\inquotes{\top}$ if and only if runs consistent with $w$ satisfy the property.
The proof is by induction on the length of $w$ in one direction, and inductively constructing $w$ using \lemref{delta} in the other; exactly like~\cite[Lem. A.4]{FMW22a}.

\begin{lemma}
\label{lem:deltastar}
$\rho \models_M \psi$ holds iff
there exist a well-formed word $w=\varsigma_1 \cdots \varsigma_{n}\in \Theta_\last^*$ that is consistent with a run $(M,\rho)$ and states 
$q_0, \dots, q_n \in Q$ such that $q_0=\inquotes{\psi}$, $q_n = \inquotes{\top}$,
and $(q_{i+1},\varsigma_{i+1}) \in \delta(q_i)$ for all $0\leq i < n$.
\end{lemma}

Before proving the main result about $\NFA$, the following technical observations are needed to show that $\lambda$ can be omitted in the NFA (cf. \cite[Lem. A.5]{FMW22a}).
Since $\TT$-equivalence causes small changes in the proof, it is given below.

\begin{lemma}
\label{lem:delta:last}
Let $\psi \in \LL_\Sigma \cup \{\top, \bot\}$ and
$(\inquotes{\top},\varsigma) \in \delta(\inquotes{\psi})$ such that $\varsigma$ is $\lambda$-consistent with step $i$ of some run $(M,\rho)$.
\begin{compactenum}
% \item[(1)] The set $\varsigma$ does not contain both $\lambda$ and $\neg\lambda$.
% \item[(2)] If $\chi$ is not equivalent to $\top$ or $\bot$ then
% $\varsigma$ contains $\last$ or $\neg \last$.
% \item[(3)] If $\last \in \varsigma$ then $\chi = \top$ or $\chi = \bot$.
\item[(1)] If $\neg \last \in \varsigma$, there is some $(\inquotes{\top},\varsigma') \in \delta(\inquotes{\psi})$
such that $\neg \last\not\in \varsigma'$ and $\varsigma'$ is consistent with step $i$ of $(M,\rho)$ as well.
\item[(2)] If $\last \in \varsigma$, there is some $(\inquotes{\top},\varsigma') \in \delta(\inquotes{\psi})$
such that $\last\not\in \varsigma'$ and $\varsigma'$ is consistent with step $i$ of $(M,\rho)$ as well.
\end{compactenum}
\end{lemma}
\begin{proof}
We first observe that the set $\varsigma$ does not contain both $\lambda$ and $\neg\lambda$, which is easy to show by induction on $\psi$: it clearly holds for the base cases, and otherwise it follows from the definition of $\owedge$ and $\ovee$.
\begin{compactenum}
\item[(1)]
We apply induction on $\psi$.
If $\psi$ is $\top$, $\bot$, or a constraint, there is nothing to show because $\neg \lambda \not\in \varsigma$.
If $\psi$ is of the form $\X\psi'$, then by definition of $\LL_\Sigma$,
$\psi'$ is not $\top$. Therefore by definition of $\delta$ the claim  holds.
The claim also holds for the tuples $\{(\inquotes{\top},\{ \lambda\}),(\inquotes{\bot},\{\neg\lambda\})\}$ used in the definition of $\delta(\G\psi')$ for some $\psi'$.
Otherwise, $(\chi,\varsigma)$ must result from $\delta(\inquotes{\psi_1})\ovee \delta(\inquotes{\psi_2})$ or $\delta(\inquotes{\psi_1})\owedge \delta(\inquotes{\psi_2})$, for some $\psi_1$ and $\psi_2$, i.e., there
must be some $(\inquotes{\chi_1},\varsigma_1) \in \delta(\inquotes{\psi_1})$
and
$(\inquotes{\chi_2},\varsigma_2) \in \delta(\inquotes{\psi_2})$ such that
$\chi \equiv_\TT  \chi_1 \wedge \chi_2$ or $\chi \equiv_\TT  \chi_1 \vee \chi_2$
and $\varsigma = \varsigma_1 \cup \varsigma_2$.
As $\neg \lambda \in \varsigma$, it must be that $\neg \lambda \in \varsigma_1$ or $\neg \lambda \in \varsigma_2$. 
Consider first the case of $\chi \equiv_\TT  \chi_1 \wedge \chi_2$, so we can assume that $\chi_1=\chi_2=\top$.
If $\neg \lambda \in \varsigma_1$, by the induction hypothesis, 
there is some $(\inquotes{\top},\varsigma_1') \in \delta(\inquotes{\psi_1})$
such that $\neg \last\not\in \varsigma_1'$ and $\varsigma_1'$ is consistent with step $i$ of $\rho$ as well. Similarly, if $\neg \lambda \in \varsigma_2$ there is some $(\inquotes{\top},\varsigma_2') \in \delta(\inquotes{\psi_2})$
such that $\neg \last\not\in \varsigma_2'$ and $\varsigma_2'$ is consistent with step $i$ of $\rho$ as well. Therefore, $(\inquotes{\top}, \varsigma_1'\cup\varsigma_2') \in \delta(\inquotes{\psi})$ satisfies the claim.

If $\chi =\top \equiv_\TT \chi_1 \vee \chi_2$, then the case that $\chi_1 = \neg \chi_2$ can be excluded because by definition of $\LL_\Sigma$ this is only possible if $\psi_1 \equiv_\TT \neg \psi_2$ are constraints, and in this case 
$\neg\lambda \not\in \varsigma$ cannot hold.
So we can assume w.l.o.g. that $\chi_1 \equiv_\TT  \top$.
By induction hypothesis, there is some $(\inquotes{\top},\varsigma_1') \in \delta(\inquotes{\psi_1})$
such that $\neg \last\not\in \varsigma_1'$ and $\varsigma_1'$ is consistent with step $i$ of $\rho$ as well.
By \lemref{delta total}, there is some 
$(\chi_2',\varsigma_2') \in \delta(\inquotes{\psi_2})$ such that
$\varsigma_2'$ is consistent with step $i$ of $\rho$, and such that $\neg \last \not\in \varsigma_2'$.
Therefore, $(\inquotes{\top}, \varsigma_1'\cup\varsigma_2') \in \delta(\inquotes{\psi})$ satisfies the claim.
\item[(2)]
This case is shown in a similar way as Item (1).
\qedhere
\end{compactenum}
\end{proof}

\noindent
Finally, \propref{nfa} can be shown as follows:

\propositionNFA*
\begin{proof}
\begin{compactenum}
\item[($\Longleftarrow$)]
If $\rho \models_M \psi$ then by \lemref{deltastar} there are a 
well-formed word 
$w = \varsigma_0 \varsigma_1 \cdots \varsigma_{n}$ in $\Theta_\lambda^*$ 
that is consistent with $\rho$
and
$q_0, \dots, q_n \in Q$ such that $q_0=\inquotes{\psi}$, $q_n = \inquotes{\top}$,
and $(q_{i+1},\varsigma_{i+1}) \in \delta(q_i)$ for all $0\leq i < n$.
\defref{NFA} filters the tuples returned by $\delta$ to some extent, but $w$ is still accepted: Indeed, as $w$ is well-formed, $\lambda \not\in\varsigma_i$ for all $i<n$.
So $\NFA$ must admit either a transition sequence
$\inquotes{\psi} \to^*_w \inquotes{\top}$, or
$\inquotes{\psi} \to^*_w q_e$, but in both cases $w$ is accepted because
$\inquotes{\psi}$ is the initial state.
\item[($\Longrightarrow$)]
Let $w = \varsigma_0 \varsigma_1 \cdots \varsigma_{n}\in \Theta^*$ be accepted by $\NFA$
via some transition sequence
$q_0 \goto{\varsigma_0} q_1 \goto{\varsigma_1} \dots \goto{\varsigma_{n}} q_{n+1}$, and such that $w$ is consistent with $(M,\rho)$.
By \defref{NFA}, there are $\varsigma_i'$, 
such that 
$\varsigma_i = \varsigma_i' \setminus \{\last, \neg \last\}$
for all $i$, $0\leq i \leq n$,
 $q_{i+1} \in \delta(q_i, \varsigma_i')$ for all $0\leq i < n$, and $\inquotes{\top} \in \delta(q_n, \varsigma_n')$.
Since $\varsigma_i$ and $\varsigma_i'$ differ only by $\lambda$ and $\neg\lambda$, the word $w' = \varsigma_0' \varsigma_1' \cdots \varsigma_{n}' \in \Theta_\lambda^*$
is consistent with $(M,\rho)$ as well.
We can moreover assume that $w'$ is well-formed:
First, by \defref{NFA}, if $\lambda \in \varsigma_i'$ for some $i<n$ then it must be $q_{i+1}=\inquotes{\top}$ by \defref{NFA}. By \lemref{delta:last} (2), there is some $\varsigma_i''$ without $\lambda$ that is consistent with step $i$ of $(M, \rho)$ as well and 
such that $q_{i+1} \in \delta(q_i, \varsigma_i'')$.
Second, if $\neg \lambda \in \varsigma_n'$ then by \lemref{delta:last} (1) there is some $\varsigma_n''$ without $\neg \lambda$ that is consistent with step $n$ of $(M, \rho)$ as well and 
such that $q_{n+1} \in \delta(q_n, \varsigma_n'')$.
Thus also the word $w'$ leads to $\top$, and $w'$ is consistent with $(M,\rho)$.
According to \lemref{deltastar}, $\rho \models_M \psi$.
\qedhere
\end{compactenum}
\end{proof}

\begin{example}
\label{exa:nfa2}
We detail the construction sketched in \exaref{nfa}, where
$\psi = (x\geq 0) \U (s\eqn \m s_2 \wedge x\eqn4)$. Let $c_0:= (x\geq 0)$, $c_1 := (x\eqn 4)$ and $c_2 := (s\eqn \m s_2)$. We have
\[
\delta(\inquotes{\psi}) = 
\delta(\inquotes{c_1 \wedge c_2}) \ovee (\delta(\inquotes{c_0}) \owedge \delta(\inquotes{\X \psi})) 
\]
Thus we compute
\begin{align*}
\delta(\inquotes{c_1 \wedge c_2}) &=\{(\inquotes{\top},\{c_1, c_2\}),(\inquotes{\bot},\emptyset)\} \\
\delta(\inquotes{c_0}) &= \{(\inquotes{\top},\{c_0\}),(\inquotes{\bot},\emptyset)\}\\
\delta(\inquotes{\X \psi})) &=\{(\inquotes{\psi},\{\neg \last\}), (\inquotes{\bot}, \{\last\})\} \\
\delta(\inquotes{c_0}) \owedge \delta(\inquotes{\X \psi}) =
\{&
(\inquotes{\psi},\{c_0, \neg  \last\}),
(\inquotes{\bot},\{c_0, \last\}),\\
&(\inquotes{\bot},\{ \neg \last\}),
(\inquotes{\bot},\{ \last\})
\} 
\end{align*}
so that
\begin{align*}
\delta(\inquotes{\psi}) = \{&
(\inquotes{\top},\{c_0, c_1, c_2, \neg  \last\}),
(\inquotes{\top},\{c_0, c_1, c_2, \last\}),\\
&(\inquotes{\top},\{ c_1, c_2,\neg \last\}),
(\inquotes{\top},\{ c_1, c_2,\last\}) \\
 &(\inquotes{\psi},\{c_0, \neg  \last\}),
(\inquotes{\bot},\{c_0, \last\}),\\
&(\inquotes{\bot},\{ \neg \last\}),
(\inquotes{\bot},\{ \last\})
\end{align*}
Moreover, $\delta(\inquotes{\top}) =\{(\inquotes{\top}, \emptyset)\}$ and 
$\delta(\inquotes{\bot}) =\{(\inquotes{\bot}, \emptyset)\}$.
For every tuple returned by $\delta$ we draw an edge in the NFA, but
by \defref{NFA}, the $\lambda$'s are dropped. So $\NFA$ is as follows:

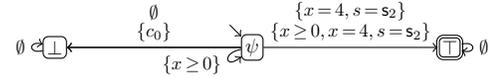
{\centering
\tikz[node distance=35mm]{
\tikzstyle{formula}=[scale=.75, rectangle, rounded corners=2pt, inner sep=2pt, draw]
\tikzstyle{edge}=[scale=.7]
\node[formula] (psi) {${\psi}$};
\node[formula, left of=psi] (bot) {${\bot}$};
\node[formula, right of=psi, double] (top) {${\top}$};
\draw[->] ($(psi) + (-.3,.3)$) to (psi);
\draw[->] (psi) to 
 node[edge, above, yshift=0pt]{$\{x\,{\geq}\,0, x\eqn 4, s\eqn \m s_2\}$}
 node[edge, above, yshift=4mm]{$\{x\eqn 4, s\eqn \m s_2\}$}
 (top);
\draw[->] (psi) to 
node[edge, above, yshift=4mm]{$\emptyset$} 
node[edge, above, yshift=0pt]{$\{c_0\}$}
(bot);
\draw[->] (psi) to (bot);
\draw[->] (top) to[loop right, looseness=6] node[edge, right]{$\emptyset$} (top);
\draw[->] (bot) to[loop left, looseness=6] node[edge, left]{$\emptyset$} (bot);
\draw[->] (psi) to[loop, out=190,in=220, looseness=7] node[edge, left, yshift=-1mm]{$\{x\,{\geq}\,0\}$}(psi);
}\par

}
We can simplify this automaton for our purposes, removing all transitions labeled $\varsigma$ such that there is another transition labeled $\varsigma' \subset \varsigma$ between the same states, since any run that is consistent with the former is also consistent with the latter.
Moreover, we can omit the sink state $\bot$, from which the accepting state is unreachable.
We thus obtain the following NFA:

{\centering
\tikz[node distance=35mm]{
\tikzstyle{formula}=[scale=.75, rectangle, rounded corners=2pt, inner sep=2pt, draw]
\tikzstyle{edge}=[scale=.7]
\node[formula] (psi) {${\psi}$};
% \node[formula, left of=psi] (bot) {${\bot}$};
\node[formula, right of=psi, double] (top) {${\top}$};
\draw[->] ($(psi) + (-.3,.3)$) to (psi);
\draw[->] (psi) to 
 node[edge, above, yshift=0pt]{$\{x\eqn 4, s\eqn \m s_2\}$}
 (top);
% \draw[->] (psi) to node[edge, above, yshift=0pt]{$\emptyset$} (bot);
% \draw[->] (psi) to (bot);
\draw[->] (top) to[loop right, looseness=6] node[edge, right]{$\emptyset$} (top);
% \draw[->] (bot) to[loop left, looseness=6] node[edge, left]{$\emptyset$} (bot);
\draw[->] (psi) to[loop, out=190,in=220, looseness=7] node[edge, left, yshift=-1mm]{$\{x\,{\geq}\,0\}$}(psi);
}\par

}
\end{example}

\section{Model Checking}

\abstraction*
\begin{proof}
\begin{compactenum}[(1)]
\item
Let $(M,\rho)$ be a witness for $\psi$, for $\rho$ of the form
$\rho \colon \alpha_0
\goto{a_1} \alpha_1
\goto{a_2} \dots
\goto{a_n} \alpha_n$.
By \propref{nfa} there is a word $w = \langle \varsigma_0, \dots, \varsigma_n\rangle$ accepted by $\NFA$
such that $(M,\alpha_i) \models \varsigma_i$ for all $0 \leq i \leq n$.
We show by induction on $i$ that $(M,\nu(\rho|_i))$ satisfies $H(\sigma(\rho|_i),w|_{i+1})$ where $\rho|_i$ is the $i$-step prefix of $\rho$ and $w|_i$ is the prefix of $w$ of length $i$.
For $i=0$ we have $H(\sigma(\rho|_0),w|_{1}) = \bigwedge_{v\in V} v_0{=}I(v) \wedge \varsigma_0(\vec V_0)$,
a formula with free variables $V_0$.
As  $\alpha_0 \models_M \varsigma_0$ by assumption, and $\nu(v_0) = \alpha_0(v)$ for all $v\in V$, the claim holds.
In the inductive step, $H(\sigma(\rho|_{i+1}),w|_{i+2}) = H(\sigma(\rho|_i),w|_{i+1}) \wedge \trans{a_{i+1}}(\vec V_i, \vec V_{i+1}) \wedge \constr(\varsigma_{i+1})(\vec V_{i+1})$.
As $\rho$ is a run, the transition assignment $\beta$ given by $\beta(v^r) = \alpha_i(v)$ and $\beta(v^w) = \alpha_{i+1}(v)$ satisfies $t_{i+1}$, and for all $v\notin \writ(a_{i+1})$, we must have $\alpha_i(v) = \alpha_{i+1}(v)$ for all $v\in V$.
As by definition $\nu(v_i) = \alpha_i(v)$ and $\nu(v_{i+1}) = \alpha_{i+1}(v)$,
$(M,\nu(\rho|_{i+1}))$ satisfies $H(\sigma(\rho|_{i+1}),w|_{i+2})$, which concludes the induction proof. The claim follows from the case $i=n$. 
\item
Suppose there are a transition sequence $\sigma$ compatible with a word $w\in \Theta^*$, a structure $M$ and an assignment $\nu$ such that $(M,\nu)$ satisfies $H(\sigma,w)$.
We show by induction on $i$ that $\rho(M, \nu, \sigma|_i)$ is a run consistent with $w|_{i+1}$.
If $i=0$ then $\sigma$ must be the empty sequence.
As $\nu \models_M H(\sigma,w)$ and $H(\sigma,w)$ contains $\bigwedge_{v\in V} v_0{=}I(v) \wedge \constr(\varsigma_0)(\vec V_0)$, by definition $\rho(M, \nu, \sigma|_0)$ must be the empty run $\alpha_0$ and compatible with $w|_1$ because we must have $\alpha_0 \models_M \constr(\varsigma_0)$.
In the inductive step, we assume that $\rho(M, \nu, \sigma|_i)$ is a run consistent with $w_{i+1}$ and consider the case of $\rho(M, \nu, \sigma|_{i+1})$.
Since $\nu \models_M \trans{a_i}(\vec V_i, \vec V_{i+1})$,
for the assignments $\alpha_i$ and $\alpha_{i+1}$ defined in $\rho(M, \nu, \sigma|_{i+1})$ the transition assignment $\beta$ given by $\beta(v^r) = \alpha_i(v)$ and $\beta(v^w) = \alpha_{i+1}(v)$ satisfies $t_i$, and for all $v\notin \writ(t_i)$ satisfies $t_{i+1}$, and the value of all variables that are not written is propagated by inertia. Thus $\alpha_i \goto{t_{i+1}} \alpha_{i+1}$ is a step, and hence $\rho(M, \nu, \sigma|_{i+1})$ is a run.
It is compatible with $w_{i+2}$ because $\rho(M, \nu, \sigma|_i)$ is consistent with $w|_{i+1}$ and $\alpha_{i+1} \models_M \varsigma_{i+1}$.
This concludes the induction proof.
Finally, it follows from \propref{nfa} that $(M,\rho)$ is a witness.
\qedhere
\end{compactenum}
\end{proof}

The next proof resembles that of \cite[Lem. 4.6]{FMW22a}, with the difference that all satisfiability and equivalence checks are performed in the model completion $\TT^*$.

\productlemma*
\begin{proof}
We prove the following two claims, corresponding to the two directions of the statement:
\begin{compactenum}
\item[(1)]
If $\smash{\NN_\BB^\psi}$ has a path $\pi$ of length $n\geq 1$ to a node $(q,\phi)$, then
$\NN_\psi$ has a transition sequence ending in $q$ labeled $w_\pi$ and
$\BB$ a transition sequence
$\sigma_\pi$ such that
$\phi \equiv_{\TT^*} \hist(\sigma_\pi,w_\pi)$ is $\TT^*$-satisfiable.
\item[(2)]
If $\NN_\psi$ has a transition sequence ending in $q$ labeled $w$ and $\BB$ a
transition sequence $\sigma$ such that
$\hist(\sigma,w)$ is $\TT^*$-satisfiable, then $\smash{\NN_\BB^\psi}$ has a path $\pi$ to $(q,\phi)$ s.t. $w=w_\pi$, $\sigma=\sigma_\pi$, and
$\phi \equiv_{\TT^*} \hist(\sigma,w)$.
\end{compactenum}
Throughout the proof we use the fact that a history constraint $\hist(\sigma,w)$ is $\TT$-satisfiable iff it is $\TT^*$-satisfiable since it is an existential formula,
cf. Rem.~\ref{rem:TT*}.
\begin{compactenum}[(1)]
\item
The proof is by induction on the length of the path $\pi$ in $\smash{\NN_\BB^\psi}$.
In the base case, $\pi$ consists only of a single step.
Then $\sigma_\pi$ is the empty transition sequence, and $w_\pi = \langle \varsigma_0\rangle$.
By definition of $\phiinit$ and the initial steps $(i)$ in \defref{product}, $\phi_1 \equiv_{\TT^*}\bigwedge_{v\in V} v{=}I(v) \wedge 
\varsigma_0 = \hist(\sigma_\pi,w_\pi)$, and this formula must be $\TT$-satisfiable. By construction of $\smash{\NN_\BB^\psi}$, there must be a step $q_0\goto{\varsigma_0}q_1$ in $\NFA$.

In the step case, let $\pi$ be of the form
\begin{equation*}
\label{eq:path2}
(q_0, \phiinit) \goto{\top,\varsigma_0} (q_1, \phi_1) \goto{t_1,\varsigma_1} \dots \goto{t_n,\varsigma_n} (q_{n+1},\phi_{n+1})
\end{equation*}
By the induction hypothesis, we assume
that $\NN_\psi$ has a transition sequence ending in $q_{n}$ labeled $w'=\langle \varsigma_0, \dots, \varsigma_{n-1}\rangle$,
$\BB$ a transition sequence
$\sigma' = \langle t_1, \dots, t_{n-1}\rangle$, and
$\varphi_n \equiv_{\TT^*} \hist(\sigma',w')$ is $\TT$-satisfiable.
Since there is the edge 
$(q_{n},\varphi_{n}) \goto{t_n,\varsigma_n} (q_{n+1},\varphi_{n+1})$, 
by \defref{product} there must be a transition
$q_{n} \goto{\varsigma_{n}} q_{n+1}$ in $\NN_\psi$, 
such that $\varphi_{n+1} \equiv_{\TT^*} \update(\varphi_{n}, t_n) \wedge \varsigma_{n}$, and
the formula $\varphi_{n+1}$ is $\TT$-satisfiable.
We thus have
\begin{align*}
\varphi_{n+1} &\equiv_{\TT^*}
\update(\varphi_{n}, t_n) \wedge \varsigma_{n} \\
& \equiv_{\TT^*} \update(\hist(\sigma',w'), t_n) \wedge \varsigma_{n}
= \hist(\sigma_\pi,w_\pi).
\end{align*}
\item 
By induction on the length $n$ of $\sigma$.
If $n\,{=}\,0$ then $\sigma$ is empty and
$w = \langle\varsigma\rangle$ for some $\varsigma \in \Theta$.
By assumption
$\hist(\sigma,w) = \bigwedge_{v\in V} v{=}I(v) \wedge \varsigma$ is $\TT$-satisfiable. 
Thus, by \defref{product} there is an $(i)$-step
$(q_0, \phiinit) \goto{\top, \varsigma} (q_1, \varphi_1)$ and we have
$\varphi_1 \equiv_{\TT^*} \phiinit \wedge \varsigma = \bigwedge_{v\in V} v{=}I(v) \wedge \varsigma$.

In the inductive step, $\sigma=\langle t_1, \dots, t_n\rangle$, and 
$w = \langle\varsigma_0 \cdots \varsigma_n\rangle$
is accepted by $\NN_\psi$ along a transition sequence $q_0 \to_w^* q_{n+1}$, such that 
$\hist(\sigma,w)$ is satisfiable.
By the induction hypothesis, $\NN_\BB^\psi$ has a path $\pi'$ to a node
$p = (q_{n},\varphi_{n})$
such that $\varphi_{n} \equiv_{\TT^*} \hist(\sigma',w')$.
We thus have that
\begin{align*}
\update(\varphi_{n}, t_n) \wedge \varsigma_n &\equiv_{\TT^*} \update(\hist (\sigma',w'), t_n) \wedge \varsigma_n \\
&= \hist(\sigma,w)
\end{align*}
is $\TT$-satisfiable.
By construction
$\NN_\BB^\psi$ must have a node $p' = ( q_{n+1}, \varphi_{n+1})$ 
such that $\varphi_{n+1} \equiv_{\TT^*} \update(\varphi_{n}, t_n) \wedge \varsigma_n$
and an edge $p \goto{t_n, \varsigma_n} p'$ can be added to $\pi'$ to obtain the desired path $\pi$.
\qedhere
\end{compactenum}
\end{proof}

\begin{example}
\label{exa:lists}
\newcommand{\xlist}{l}
\newcommand{\xel}{x}
\newcommand{\nil}{[]}
\newcommand{\consop}{\mathrel{::}}
\newcommand{\car}{\m{hd}}
\newcommand{\cdr}{\m{tl}}
Suppose a (part of a) business process manages tasks by inserting and extracting them in a stack.
It can be modeled as a \mydds over the theory of acyclic lists \cite{Oppen80},
with sorts $\mathit{list}$ and $\mathit{task}$, unary function symbols $\car$, $\cdr$, binary $\consop$ and constant $\nil$. 
As variables, the \plainmydds uses $\xlist$ of type $list$,
% The working memory of the \plainmydds is modeled by a variable $\xlist$ of type $list$ 
$\xel$ of sort $\mathit{task}$, and a status variable $s$ of uninterpreted sort $\mathit{status}$, and values in $\{\m{init}, \m{pop}, \m{push}, \m{end}\}$. 
We assume that $\mathit{status}$  is
written by the parent process to determine the next operation. 
Let $I(s)=\m{init}$, and $I(\xlist)=\nil$ and $I(\xel)=\m a$.
Transitions are as follows, where $u$ is a %fixed 
term of type $\emph{list}$, e.g. $u =\m a \consop\m b \consop \m c \consop \nil$, for $\m a$, $\m b$, $\m c$ %constants 
of type $\mathit{task}$:

{\centering
\begin{tabular}{r@{ }l}
$\m{initialize} =$ & $(s^r\eqn\m{init} \land \xlist^w = u \land \exists d.s^w=d)$\\
$\m{finalize} =$ & $(s^r\eqn\m{end})$\\
$\m{do\_push} =$ & $(s^r\eqn\m{push} \land \xlist^w= \xel^w \consop \xlist^r \land  \exists d.s^w=d)$ \\
$\m{do\_pop} =$ & $(s^r\eqn\m{pop} \land \xlist^r \neq \nil \land \xlist^w= \cdr(\xlist^r)\land {}$ \\
&\qquad$\xel^w=\car(\xlist^r) \land \exists d.s^w=d)$ %\\
% $\m{finalize}$ & $(s^r\eqn\m{end})$
\end{tabular}
% 
% 
% \begin{tikzpicture}[node distance=80mm]
% \node[state] (1) {$1$};
% \node[state,right of = 1, double] (2)  {$2$};
% \node[state, right of = 2, double] (3) {$3$};
% \draw[edge] (1) to node[action, above=0mm] {$\m{init}\colon [x_{s}^w\neq \text{end} \land \xlist^w= t]$} (2);
% \draw[edge] (2) to[loop, out=120,in=60, looseness=10] node[action, above, yshift=-1mm]{$\m{pop}\colon [x_{s}^r=\text{pop} \land \xlist^r \neq \text{nil} \land \xlist^w= \text{cdr}(\xlist^r)\land \xel^w=\text{car}(\xlist^r)\land \exists d.\,x_{s}^w=d]$} (2);
% \draw[edge] (2) to[loop below, out=-120,in=-60, looseness=10] node[action, below,yshift=1mm]{$\m{push}\colon [x_{s}^r=\text{push} \land \xlist^w= \text{cons}(\xel^w,\xlist^r) \land \exists d.x_{s}^w=d]$} (2);
% \draw[edge] (2) to node[action, above=0mm] {$\m{end}\colon [x_{status}^r=\text{end}]$}(3);
% \end{tikzpicture}
\par
}

\noindent
Since $\TT$ has QE~\cite{Mal62},
\thmref{model:checking} applies: e.g. one can check that the process may reach the undesirable state where the list is empty but the next operation is $\text{pop}$, i.e., there is a witness for
$\F (\xlist\,{=}\,\m{nil} \wedge \X (s=\m{pop}))$.
\end{example}

\section{Decidability Criteria}

\paragraph{I. Acyclic signature.}
Below is an example of a \DBDDS over an acyclic signature, which is thus in a decidable class by \thmref{acyclic:SAS}.
\begin{example}
\label{exa:acyclic}
\newcommand{\eps}{\epsilon}
The following example models a collaborative incident management in a software company, which has first- and second-level support, as well as support by developers. 
The example is adapted from the VERIFAS problem set~\cite{verifas}\footnote{\emph{https://github.com/oi02lyl/has-verifier/blob/master/ bpmn/Incident-Management-as-Collaboration.txt}}, omitting updates of relations.
We consider EUF over the signature $\Sigma$ with sorts $\SS = \{\mathit{string}, \mathit{cust\_id}, \mathit{prob\_id}, \mathit{agent\_id}, \mathit{status}\}$.
There are 
$\mathit{string}$ constants $\eps$, $\mathtt{Yes}$, $\mathtt{No}$, $\mathtt{Low}$, $\mathtt{1}$, and $\mathtt{2}$,
constants $\m{start}$, $\m{problemReceived}$, $\m{handling}$, $\m{solved}$, and  $\m{unsolved}$
of type status, and
constant $\m{NULL}$ of sort $\mathit{prob\_id}$ to indicate undefinedness, and the following relations:
$\m{Customer} \subseteq \mathit{cust\_id} \times \mathit{string}$ contains the ids and names of customers;
$\m{ProblemType}\subseteq \mathit{prob\_id} \times \mathit{string}\times \mathit{string}$ contains ids and names of problem types as well as their importance; and
$\m{SupportAgent}\subseteq \mathit{agent\_id} \times \mathit{string}\times \mathit{string}$ contains ids of agents, whether they are developers, and their support level.
In addition, there is a function $\m{agentType}\colon \mathit{agent\_id} \to \mathit{prob\_id}$ mapping the id of an agent to the type of problems they can handle.
Note that the sort graph has only a single edge $\mathit{agent\_id} \to \mathit{prob\_id}$, so the signature is acyclic.

The following variables are used in the \DBDDS: 
$p$ of type $\mathit{string}$, the problem;
$cid$ of type $\mathit{cust\_id}$, the id of the current customer;
$pid$ of type $\mathit{prob\_id}$, the id of the current problem;
$sol$ of type $\mathit{string}$, the solution; and
$s$ of type $\mathit{status}$.
The control structure of the \DBDDS is as follows, with the status being shown as a control state.
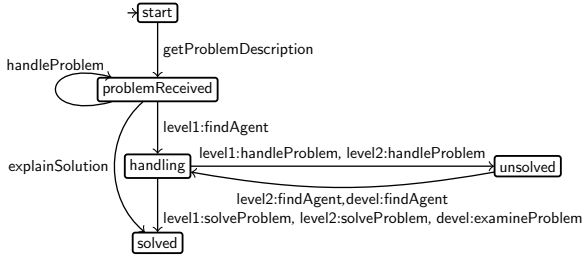
\begin{center}
\begin{tikzpicture}[node distance=17mm]
\tikzstyle{state}=[draw, rectangle, rounded corners=1pt, inner sep=3pt, line width=.7pt, scale=.6]
\tikzstyle{action}=[scale=0.6]
\tikzstyle{label}=[scale=0.6]
\node[state] (0)  {$\m{start}$};
\node[state, below of=0] (1) {$\m{problemReceived}$};
\node[state, below of=1] (2) {$\m{handling}$};
\node[state, right of=2, xshift=65mm] (3) {$\m{unsolved}$};
\node[state, below of=2] (4) {$\m{solved}$};
\draw[edge] ($(0) + (-.4,0)$) -- (0);
\draw[edge] (0) to node[label,right] {$\m{getProblemDescription}$} (1);
\draw[edge] (1) to node[label,right] {$\m{level1{:}findAgent}$} (2);
\draw[edge] (2) to node[label,above] {$\m{level1{:}handleProblem}$, $\m{level2{:}handleProblem}$} (3);
\draw[edge] (2) to node[label,right, near end] {$\m{level1{:}solveProblem}$, $\m{level2{:}solveProblem}$, $\m{devel{:}examineProblem}$} (4);
\draw[edge, loop left] (1) to node[label,above, yshift=3mm] {$\m{handleProblem}$} (1);
\draw[edge, bend left=10] (3) to node[label,below] {$\m{level2{:}findAgent}$,$\m{devel{:}findAgent}$} (2);
\draw[edge, bend right=50] (1) to node[label,left] {$\m{explainSolution}$} (4);
\end{tikzpicture}
\end{center}
The transition guards are as follows:\\
\tikz{\node[scale=.8]{$
\begin{array}{@{}r@{\,}l@{}}
\m{getProblemDescription} & =
(\exists n. \m{Customer}(cid^w, n) \wedge p^w \neq \eps \wedge {}\\
&\qquad pid^w \neq \m{NULL} \wedge sol^w = \eps)
\\
\m{handleProblem} &= (\exists n. \m{ProblemType}(pid^r, n, \mathtt{Low}) \wedge sol^w \neq \eps)
\\
\m{explainSolution} &=(sol^r \neq \eps) 
\\
\m{level1{:}findAgent} &= (sol^r  = \eps \wedge \m{SupportAgent}(aid^w, \mathtt{No}, \mathtt{1}) \wedge {}\\&\quad \m{agentType}(aid^w) = pid^r \wedge aid^w \neq \m{NULL})
\\
\m{level1{:}handleProblem} & =(sol^r = \eps \wedge sol^w = \eps) 
\\
\m{level1{:}solveProblem} & =(sol^r = \eps \wedge sol^w \neq \eps) 
\\
\m{level2{:}findAgent} &= (sol^r  = \eps \wedge \m{SupportAgent}(aid^w, \mathtt{No}, \mathtt{2})\wedge {}\\&\quad \m{agentType}(aid^w) = pid^r \wedge aid^w \neq \m{NULL})
\\
\m{level2{:}handleProblem} & =(sol^r = \eps \wedge sol^w = \eps) 
\\
\m{level2{:}solveProblem} & =(sol^r = \eps \wedge sol^w \neq \eps) 
\\
\m{devel{:}findAgent} &=
(\exists l. \m{SupportAgent}(aid^w, \mathtt{Yes},l)\wedge {}\\&\quad \m{agentType}(aid^w) = pid^r \wedge aid^w \neq \m{NULL})
\\
\m{devel{:}examineProblem} &=(pid^r \neq \m{NULL} \wedge sol^w \neq \eps)
\\
\end{array}
$}}
One can e.g. check that if a problem is solved without consulting an agent, then its importance is low, i.e.
there is no witness for
$\F (\m{solved} \wedge aid{=}\m{NULL} \wedge \exists n\: i. (\m{ProblemType}(pid, n, i) \wedge i \neq \mathtt{Low}))$.
\end{example}

\paragraph{II. Acyclic signature and monotonicity constraints.}

\theoremacyclicMC*
\begin{proof}
By~\cite[Sec. 8]{CalvaneseGGMR22}, the procedure $\m{TameCombCover}$ produces a cover for the combined theory, i.e. a quantifier-eliminated formula equivalent in the theory completion $\TT^*$.
This formula is of the form $\phi_1(\vec X) \wedge \phi_2(\vec Y, \vec t(\vec X))$ $(\dagger)$, where $X$ are the EUF variables in $V$, $\phi_1(\vec X)$ is an EUF formula, $\vec t(\vec X)$ is a list of EUF terms, $Y$ are the arithmetic variables in $V$, and $\phi_2(\vec Y, \vec Z)$ is an arithmetic formula  (cf.~\cite[{Eq. ({31})}]{CalvaneseGGMR22}).
We show that finitely many formulas of the form $(\dagger)$ exist, up to equivalence, so that a finite history set exists. 
As $\Sigma$ is acyclic, there are only finitely many non-equivalent formulas for $\phi_1(\vec X)$, and finitely many (say, $k$) terms $\vec t(\vec X)$.
By construction, $\phi_2(\vec Y, \vec Z)$ is the result of a quantifier elimination procedure in the arithmetic theory. Thus, it is
a formula where all atoms are MCs over variables $Y$ and $Z$ and where the set of constants is bounded by $\BB$ and $\psi$, cf. the proofs of \cite[Thm. 5.2]{FMW22a} and \cite[Thm. 4]{FMW22c}.
% The proof of
% \cite[Thm. 5.2]{FMW22a} shows that arithmetic DDSs over MCs have a finite history set because there are only finitely many non-equivalent formulas over MC atoms over a finite set of variables.
Since $Y \subseteq V$ is finite, and there are at most $k$ additional variables in $Z$, so by a reasoning as in \cite[Thm. 5.2]{FMW22a} (for MCs) or \cite[Thm. 4]{FMW22c} (for IPCs), the set of possible formulas for $\phi_2(\vec Y, \vec Z)$ is finite up to equivalence. 
Thus $(\BB, \psi)$ has a finite history set.
\end{proof}

\begin{restatable}{theorem}{theoremlocallyfiniteMC}
\label{thm:MC:locally:finite}
Let $\TT$ be a  tame combination of $\TT_{db}$ satisfying the assumptions of \corref{locally:finite}, and LRA.
If $\BB$ is a \DBDDS with arithmetic and $\psi \in \LL_\Sigma$ s.t. all arithmetic constraints in $\BB$ and $\psi$ are MCs, the $\TT$-verification task is decidable.
\end{restatable}
\begin{proof}
The proof coincides with that of \thmref{MC:acyclic}, except that one uses local finiteness to obtain that there are finitely many non-equivalent formulas for $\phi_1(\vec X)$ and finitely many terms $\vec t(\vec X)$.
\end{proof}

\paragraph{IV. Bounded lookback.}
The \emph{bounded lookback} property
intuitively captures that a \plainmydds maintains only a bounded amount of information, which ultimately implies decidability~\cite{FMW22a}.
It properly generalizes the \emph{feedback freedom} property~\cite{DDV12}.
Both properties are formally defined via computation graphs~\cite{DDV12}, which we adapt here to a {\DBDDS} $\BB$ with variables $V$, and $\psi \in \LL_\Sigma$. Let $C$ be the set of constraints in $\psi$.

\begin{definition}
Let $\sigma = \langle t_1, \dots, t_n\rangle$ be a transition sequence of $\BB$, and $w= \langle \varsigma_0, \dots, \varsigma_n\rangle\in 2^C$.
The \emph{computation graph} for $\sigma$ and $w$ is an undirected graph
$G_{\sigma,w}$ with node set
$\mc V = \{v_i \mid v\in V\text{, }0\,{\leq}\,i\,{\leq}\,n\}$
and an edge $(x_i, y_j)$ whenever $x_i$ and $y_j$ are in the transitive closure of variable pairs that occur in a common literal of 
$H(\sigma,w)$, for all for $i,j \leq n$.
% $\trans{t_{k}}(\vec V_{k-1}, \vec V_k)$ for $1\leq k \leq n$ or $\varsigma_i(\vec V_k)$ for $0\leq k \leq n$.
\end{definition}

For instance, there is an edge from $x_i$ to $y_j$ if $H(\sigma,w)$ contains a literal $x_i < y_j + 3$, $x_i = y_j$, or $\exists z.\: (x_i \eqn z \wedge z \eqn y_j)$.
The subgraph $E_{\sigma,w}$  of $G_{\sigma,w}$ consists of all edges $(x_i, y_j)$ such that $x_i$ and $y_j$ are in the transitive closure of variable pairs that occur in a common \emph{equality} atom of $H(\sigma,w)$.
E.g., the latter two cases above imply an edge $(x_i, y_j)$ in $E_{\sigma,w}$, called \emph{equality edge}.
The graph obtained from $G_{\sigma,w}$ by collapsing equality edges is denoted $[G_{\sigma,w}]$.

\begin{example}
\label{exa:computation:graphs}
Consider the two transition sequences $\sigma_1 = \langle \m{login}, \m{select}, \m{sum\_up},\m{discount}, \m{restart}, \m{select},$ $\m{sum\_up},\m{discount}, \m{ship}\rangle$ for the \DBDDS in \exaref{intro};
and $\sigma_2 = \langle \m{xset}, \m{yset}, \m{xset}, \m{yset}, \m{xset}\rangle$ for the system in \exaref{simple}. For simplicity, let $w = \langle \emptyset, \dots, \emptyset\rangle$. The respective computation graphs are as follows,
 where edges in $E_{\sigma_i,w}$ (i.e., equality edges) are dotted, and all non-equality edges in $G_{\sigma,w}$ solid:
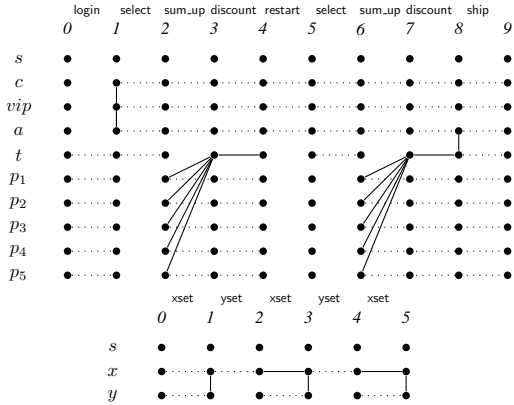
\begin{center}
\begin{tikzpicture}[xscale=.65, yscale=.8]
\node[scale=.7] at (-1,.5) {$s$};
\node[scale=.7] at (-1,.1) {$c$};
\node[scale=.7] at (-1,-.7) {$a$};
\node[scale=.7] at (-1,-1.1) {$t$};
\node[scale=.7] at (-1,-.3) {$vip$};
\node[scale=.7] at (-1,-1.5) {$p_1$};
\node[scale=.7] at (-1,-1.9) {$p_2$};
\node[scale=.7] at (-1,-2.3) {$p_3$};
\node[scale=.7] at (-1,-2.7) {$p_4$};
\node[scale=.7] at (-1,-3.1) {$p_5$};
\foreach \i in {0,1,2,3,4,5,6,7,8,9} {
  \node[scale=.65] at (\i,1) {\i};
  \node[fill, circle, inner sep=0pt, minimum width=1mm] (s\i) at (\i,.5) {};
  \node[fill, circle, inner sep=0pt, minimum width=1mm] (c\i) at (\i,.1) {};
  \node[fill, circle, inner sep=0pt, minimum width=1mm] (a\i) at (\i,-.7) {};
  \node[fill, circle, inner sep=0pt, minimum width=1mm] (t\i) at (\i,-1.1) {};
  \node[fill, circle, inner sep=0pt, minimum width=1mm] (vip\i) at (\i,-.3) {};
  \node[fill, circle, inner sep=0pt, minimum width=1mm] (p1\i) at (\i,-1.5) {};
  \node[fill, circle, inner sep=0pt, minimum width=1mm] (p2\i) at (\i,-1.9) {};
  \node[fill, circle, inner sep=0pt, minimum width=1mm] (p3\i) at (\i,-2.3) {};
  \node[fill, circle, inner sep=0pt, minimum width=1mm] (p4\i) at (\i,-2.7) {};
  \node[fill, circle, inner sep=0pt, minimum width=1mm] (p5\i) at (\i,-3.1) {};
  }
\node[scale=.5] at (.5,1.3) {$\mathsf{login}$\phantom{$\m{pl}$}};
\node[scale=.5] at (1.5,1.3) {$\mathsf{select}$\phantom{$\m{pl}$}};
\node[scale=.5] at (2.5,1.3) {$\mathsf{sum\_up}$\phantom{$\m{pl}$}};
\node[scale=.5] at (3.5,1.3) {$\mathsf{discount}$\phantom{$\m{pl}$}};
\node[scale=.5] at (4.5,1.3) {$\mathsf{restart}$\phantom{$\m{pl}$}};
\node[scale=.5] at (5.5,1.3) {$\mathsf{select}$\phantom{$\m{pl}$}};
\node[scale=.5] at (6.5,1.3) {$\mathsf{sum\_up}$\phantom{$\m{pl}$}};
\node[scale=.5] at (7.5,1.3) {$\mathsf{discount}$\phantom{$\m{pl}$}};
\node[scale=.5] at (8.5,1.3) {$\mathsf{ship}$\phantom{$\m{pl}$}};
% login
\draw (c1) -- (a1) -- (vip1);
\draw[dotted] (p10) -- (p11);
\draw[dotted] (p20) -- (p21);
\draw[dotted] (p30) -- (p31);
\draw[dotted] (p40) -- (p41);
\draw[dotted] (p50) -- (p51);
\draw[dotted] (t0) -- (t1) -- (t2);
% select
\draw[dotted] (c1)  -- (c9);
\draw[dotted] (vip1) -- (vip9);
\draw[dotted] (a1) -- (a9);
% sumup
\draw (p12) -- (t3);
\draw (p22) -- (t3);
\draw (p32) -- (t3);
\draw (p42) -- (t3);
\draw (p52) -- (t3);
\draw[dotted] (p12) -- (p14);
\draw[dotted] (p22) -- (p24);
\draw[dotted] (p32) -- (p34);
\draw[dotted] (p42) -- (p44);
\draw[dotted] (p52) -- (p54);
% discount
\draw (t4) -- (t3);
% select
\draw[dotted] (t5) -- (t6);
% sum up
\draw (p16) -- (t7);
\draw (p26) -- (t7);
\draw (p36) -- (t7);
\draw (p46) -- (t7);
\draw (p56) -- (t7);
\draw[dotted] (p16) -- (p19);
\draw[dotted] (p26) -- (p29);
\draw[dotted] (p36) -- (p39);
\draw[dotted] (p46) -- (p49);
\draw[dotted] (p56) -- (p59);
% discount
\draw (t7) -- (t8);
% ship
\draw (a8) -- (t8);
\draw[dotted] (t8) -- (t9);
\end{tikzpicture}
\begin{tikzpicture}[xscale=.65, yscale=.8]
\node[scale=.7] at (-1,.5) {$s$};
\node[scale=.7] at (-1,.1) {$x$};
\node[scale=.7] at (-1,-.3) {$y$};
% \node[scale=.7] at (-1,1) {$\sigma_2\colon$};
% \node[scale=.75] at (-1,1.3) {instant};
\foreach \i in {0,1,2,3,4,5} {
  \node[scale=.65] at (\i,1) {\i};
%   \node[scale=.65] (state\i) at (\i,1) {$\m\l$};
  \node[fill, circle, inner sep=0pt, minimum width=1mm] (s\i) at (\i,.5) {};
  \node[fill, circle, inner sep=0pt, minimum width=1mm] (x\i) at (\i,.1) {};
  \node[fill, circle, inner sep=0pt, minimum width=1mm] (y\i) at (\i,-.3) {};
  }
\node[scale=.5] at (.5,1.3) {$\mathsf{xset}$\phantom{p}};
\node[scale=.5] at (1.5,1.3) {$\mathsf{yset}$\phantom{p}};
\node[scale=.5] at (2.5,1.3) {$\mathsf{xset}$\phantom{p}};
\node[scale=.5] at (3.5,1.3) {$\mathsf{yset}$\phantom{p}};
\node[scale=.5] at (4.5,1.3) {$\mathsf{xset}$\phantom{p}};
\draw[dotted] (x0) -- (x1);
\draw[dotted] (y0) -- (y1);
\draw (x1) -- (y1);
\draw[dotted] (x1) -- (x2);
\draw (x2) -- (x3);
\draw[dotted] (y2) -- (y3);
\draw (x3) -- (y3);
\draw[dotted] (x3) -- (x4);
\draw (x4) -- (x5);
\draw[dotted] (y4) -- (y5);
\draw (x5) -- (y5);
\end{tikzpicture}
\end{center}
The computation graphs thus visualize dependencies between variables along the transition sequences.
Note that after collapsing equality edges, i.e. when considering $[G_{\sigma_i,w}]$, the longest path in the graph for $\sigma_1$  has length 5 (e.g. from $c_1$ to $p_{1,6}$);
while for $\sigma_2$ length 4 (from $x_0$ to $y_5$).
\end{example}

\noindent
We can now define bounded lookback.

\begin{definition}
For any $k\geq 0$, a
\DBDDS $\BB$ has \emph{$k$-bounded lookback} with respect to $\psi$ if for all transition sequences $\sigma$ of $\BB$ and words $w\in 2^C$ such that $H(\sigma,w)$ is $\TT$-satisfiable, all acyclic paths in $[G_{\sigma,w}]$ have length at most $k$.
\end{definition}

For instance, as the graph for $\sigma_1$ in \exaref{computation:graphs} suggests, \exaref{intro} has 5-bounded lookback with respect to all properties $\psi$ that do not compare variables with each other (e.g. $\F (s{=}\m{shipped})$):
indeed, it also becomes obvious when staring at the picture in \exaref{intro} that 
at any point of the process, current values of variables depend on at most five steps in the past. However, the system is not feedback free  as feedback freedom forbids self-referential updates such as in $\m{discount}$.

On the other hand, the transitions $\setx, \sety$ in $\sigma_2$ in \exaref{computation:graphs} can be repeated an arbitrary number of times, obtaining arbitrarily long paths even after collapsing equality edges. Thus \exaref{simple} does not have $k$-bounded lookback, for any $k$.
\smallskip

For a given $k$, it is decidable whether $\BB$ has $k$-bounded lookback~\cite[Sec.~5]{FMW22a}, by enumerating all transition sequences of $\BB$ combined with all sequences of atoms in $\psi$ such that $H(\sigma,w)$ is satisfiable, and checking that no chains of variable dependencies have length more than $k$ (not counting steps where variables are compared by equality, or assigned using the equality operator).
The following proof is similar to that of~\cite[Thm.~5.10]{FMW22a}, and repeated here only for self-containedness.

% \begin{restatable}{theorem}{theoremboundedlookback}
% \label{thm:bounded:lookback}
% If a \DBDDS $\BB$ with arithmetic has EUF as $\TT_{db}$ and LRA as $\TT_{ar}$, and $\BB$ together with a $\Sigma$-property $\psi$ has $k$-bounded lookback for some $k\geq 0$, the verification task is decidable.
% \end{restatable}
\theoremboundedlookback*

\begin{proof}
Let $C_{ext}$ be the set of all $\Sigma$-constraints in $\psi$, transitions of $\BB$, and
$\phiinit$.
% and suppose the maximal quantifier depth of a constraint in $C_{ext}$ is $q$.
Let $\vec U_1, \dots, \vec U_k$ be $k$ fresh copies of $\vec V$, and $\Psi$ be the set of $\Sigma$-formulas with free variables $V$, quantifier depth at most $k$ and where all atoms are of the form $c(\vec X)$, where $c\in C_{ext}$ and the variable vector $\vec X$ contains only variables in $\vec U_1, \dots, \vec U_k$ and $V$.
As a set of formulae with bounded quantifier depth over a finite set of atoms instantiated with finitely many variables, $\Psi$ is finite up to equivalence.
We prove that $\Psi$ is a history set.

To that end,
we show that for every history constraint $\hist(\sigma,w)$ of $\BB$ and $\psi$
there is some $\phi \in \Psi$ with $\phi \equiv_{\TT^*} \hist(\sigma,w)$.
The proof is by induction of $\sigma$.
If $\sigma$ is empty,  $\hist(\sigma,w) = \phiinit \wedge c$ for some $c\in C_{ext}$. This formula is quantifier-free, has free variables $V$, and all atoms are taken from the vocabulary $C_{ext}$. Thus $\hist(\sigma,w) \in \Psi$.
Otherwise, let $\sigma = \langle t_1, \dots, t_n\rangle$ for some $n\geq 1$.
If $\hist(\sigma, w)$ (and hence $H(\sigma,w)$) is $\TT$-unsatisfiable (hence $\TT^*$-unsatisfiable), then $\bot \equiv_{\TT^*} \hist(\sigma,w)$ and $\bot\in \Psi$.
Otherwise,
$\hist(\sigma, w) = \update(\hist(\sigma', w'), t_n) \wedge c$ for some $t_n\in T$ and $c\in C_{ext}$, and $w'$, $\sigma'$ the prefixes of $w$ and $\sigma$ without the last symbol/the last transition.
By induction hypothesis there is some $\phi' \in \Psi$ with $\phi' \equiv_{\TT^*} \hist(\sigma', w')$.
We can thus write $\hist(\sigma, w)$ as 
$\phi := \exists \vec U. \phi'(\vec U) \wedge \chi$ 
for some quantifier free formula $\chi$ whose atoms are in $C_{ext}$ and have variables $V$.
Let $[\phi]$ be obtained from $\phi$ by eliminating all equality literals $x=y$ in $\phi$ and uniformly substituting all variables in an equivalence class by a representative.
Since $(\BB,\psi)$ has $k$-bounded lookback, and $[\phi]$ encodes a part of
$[G_{\sigma,w}]$, 
$[\phi]$ is equivalent to some formula $\phi''$ of quantifier depth at most $k$ over the same vocabulary, 
obtained by dropping irrelevant literals and existential quantifiers of variables which are not connected to variables of the last instant in $[G_{\sigma,w}]$, and renaming quantified variables to variables in $U$. So $\phi''$ must be in $\Psi$.
\end{proof}

\thmref{bounded:lookback} strictly generalizes~\cite{DDV12} since bounded lookback generalizes feedback freedom~\cite[Lem. 5.11]{FMW22a}.
Note that in contrast to \thmref{acyclic:SAS}, \thmref{bounded:lookback} imposes no restriction on the signature, only on the control flow.

\end{document}